\documentclass[12pt,oneside,english]{amsart}
\usepackage[T1]{fontenc}
\usepackage[latin9]{inputenc}
\usepackage{geometry}
\usepackage{float}
\usepackage{amsthm}
\usepackage{amsbsy}
\usepackage{amstext}
\usepackage{amssymb}
\usepackage{graphicx}
\usepackage{setspace}
\usepackage{esint}
\usepackage[all]{xy}
\makeatletter

\providecommand{\tabularnewline}{\\}

\numberwithin{equation}{section}
\numberwithin{figure}{section}
\theoremstyle{plain}
\newtheorem{thm}{\protect\theoremname}
  \theoremstyle{plain}
  \newtheorem{prop}[thm]{\protect\propositionname}
  \theoremstyle{definition}
  \newtheorem{defn}[thm]{\protect\definitionname}
  \theoremstyle{plain}
  \newtheorem{cor}[thm]{\protect\corollaryname}
  \theoremstyle{plain}
  \newtheorem{lem}[thm]{\protect\lemmaname}

\usepackage{tikz}\usetikzlibrary{matrix}
\usetikzlibrary{decorations.pathreplacing}

\usepackage{hyperref}

\usepackage{babel}
\providecommand{\corollaryname}{Corollary}
  \providecommand{\definitionname}{Definition}
  \providecommand{\lemmaname}{Lemma}
  \providecommand{\propositionname}{Proposition}
\providecommand{\theoremname}{Theorem}

\makeatother

\usepackage{babel}
  \providecommand{\corollaryname}{Corollary}
  \providecommand{\definitionname}{Definition}
  \providecommand{\lemmaname}{Lemma}
  \providecommand{\propositionname}{Proposition}
\providecommand{\theoremname}{Theorem}

\begin{document}

\title{Discrete Nahm equations for $\text{SU}(N)$ hyperbolic monopoles }

\author{Joseph Y C Chan}

\address{Department of Mathematics and Statistics, University of Melbourne,
Vic 3010, Australia}

\email{jchan3@student.unimelb.edu.au}

\date{29 June 2015}

\begin{abstract}
In a paper of Braam and Austin, $\text{SU}(2)$ magnetic
monopoles in hyperbolic space $H^{3}$ were shown to be the same as
solutions to matrix-valued difference equations called the discrete
Nahm equations. Here, I discover the $(N-1)$-interval discrete Nahm equations and show that their solutions are equivalent to 
$\text{SU}(N)$ hyperbolic monopoles. These discrete time evolution
equations on an interval feature a jump in matrix dimensions at certain
points in the evolution, which are given by the mass data of the corresponding
monopole. I prove the correspondence with higher rank hyperbolic monopoles
using localisation and Chern characters. I then prove that the monopole
is determined up to gauge transformations by its ``holographic image''
of $\text{U}(1)$ fields at the asymptotic boundary of $H^{3}$. 
\end{abstract}
\maketitle


\section{Outline}

The Nahm equations are the following system of ODE 
\[
\dfrac{d\left(\sigma+\sigma^{\ast}\right)}{dt}=\left[\sigma,\sigma^{\ast}\right]+\left[\tau,\tau^{\ast}\right]
\]
\[
\dfrac{d\tau}{dt}=\left[\sigma,\tau\right]
\]
where $\sigma$ and $\tau$ are complex-valued $k\times k$ matrices,
$k\in\mathbb{N}$ and $t\in[-p,p]$, $p\in\mathbb{Z}\text{ or }\frac{1}{2}+\mathbb{Z}$.
The solutions of the Nahm equations are in one-to-one correspondence
with $\text{SU}(2)$ magnetic monopoles in $\mathbb{R}^{3}$ of mass
$p$ and charge $k$ \cite{key-1}.

\begin{table}[H]
\begin{tabular}{|c|c|c|}
\cline{2-3} 
\multicolumn{1}{c|}{} & {\small $\text{SU}(2)$ magnetic monopoles}  & {\small $\text{SU}(N)$ magnetic monopoles}\tabularnewline
\hline 
{\small Euclidean $\mathbb{R}^{3}$}  & {\small Nahm equations}  & {\small $(N-1)$-interval Nahm equations}\tabularnewline
\hline 
{\small Hyperbolic $H^{3}$}  & {\small discrete Nahm equations}  & {\small $(N-1)$-interval discrete Nahm equations}\tabularnewline
\hline 
\end{tabular}

\caption{Monopoles and Nahm equations}
\end{table}

Hurtubise and Murray \cite{key-2} discovered what I call $(N-1)$-interval
Nahm equations for $\text{SU}(N)$ magnetic monopoles in $\mathbb{R}^{3}$.
The $(N-1)$-interval Nahm equations resemble the Nahm equations on
intervals $\left[p_{1},p_{2}\right],\ldots,\left[p_{N-1},-p_{N}\right]$
where $p_{1},\ldots,p_{N}\in\mathbb{Z}\text{ or }\frac{1}{2}+\mathbb{Z}$.
Across each boundary $t=p_{i}$ for some $i\in\left\{ 1,\ldots,N-1\right\} $,
the matrices $\sigma,\tau$ change dimensions from $\left(k_{1}+\ldots+k_{i-1}\right)\times\left(k_{1}+\ldots+k_{i-1}\right)$
to $\left(k_{1}+\ldots+k_{i}\right)\times\left(k_{1}+\ldots+k_{i}\right)$.
$\sigma$ and $\tau$ have a simple pole at each boundary and their
residue at a pole is a representation of $\text{SU}(2)$.

Braam and Austin \cite{key-3} then found the discrete Nahm equations
\[
\left[\beta_{i+\frac{1}{2}},\beta_{i+\frac{1}{2}}^{\ast}\right]+\gamma_{i+1}\gamma_{i+1}^{\ast}-\gamma_{i}^{\ast}\gamma_{i}=0
\]
\[
\beta_{i-\frac{1}{2}}\gamma_{i}-\gamma_{i}\beta_{i+\frac{1}{2}}=0
\]
where $\beta_{i}$ and $\gamma_{i}$ are complex-valued $k\times k$
matrices and $i\in\left\{ -p,-p+1,\ldots,p-1,p\right\} $, $p\in\mathbb{Z}\text{ or }\frac{1}{2}+\mathbb{Z}$
(Notably, Braam and Austin only treat the half-integer case). The
solutions to the discrete Nahm equations are in one to one correspondence
with $\text{SU}(2)$ magnetic monopoles in hyperbolic 3-space $H^{3}$.

In this paper, I introduce the $(N-1)$-interval discrete Nahm equations
whose solutions are in one-to-one correspondence with (framed) $\text{SU}(N)$
magnetic monopoles in hyperbolic space. As in the continuous case,
the $(N-1)$-interval discrete Nahm equations resemble discrete Nahm
equations on $(N-1)$ intervals and at each boundary between adjacent
intervals, the matrices $\beta_{i}$ and $\gamma_{i}$ jump in dimensions.
As far as I am aware, this is the first time that this change of dimensions
behaviour has been found in a system of matrix difference equations.

Atiyah showed that hyperbolic magnetic monopoles are $S^{1}$-invariant
instantons on $\mathbb{R}^{4}$ \cite{key-4}. The $(N-1)$-interval
discrete Nahm equations arise from the ADHM construction applied to
$S^{1}$-invariant instantons. The matrices $\beta_{i}$ and $\gamma_{i}$
are found to be the block matrices within the ADHM matrices equivariant
with respect to the induced $S^{1}$ action. The $(N-1)$-interval
discrete Nahm equations are then the ADHM equations restricted to
these equivariant blocks.

The $(N-1)$-interval discrete Nahm equations can be interpreted as
the discrete evolution of block matrices within the ADHM matrices.
The solution matrices at a boundary are to be thought of as boundary
data for the evolution equations.

Atiyah also proved that there is an isomorphism between the moduli
of monopoles and the moduli of rational maps \cite{key-4,key-5}.
I produce explicit formulae for the rational map of an $\text{SU}(N)$
hyperbolic monopole in terms of the boundary data of a solution of
the $(N-1)$-interval discrete Nahm equations.

Finally, Braam and Austin \cite{key-3} showed that the boundary data
of an $\text{SU}(2)$ hyperbolic monopole was equivalent with the
boundary data in the sense of discrete Nahm equations and so determined
the monopole (up to gauge equivalence). The proof of the analogous
theorem for the $\text{SU}(N)$ case follows the same lines. However,
it is notable that the generalisation of the map 
\[
\mathbb{P}^{1}\rightarrow\mathbb{P}^{k}
\]
which appears in Braam and Austin's theorem generalises to $\left(N-1\right)$
maps from $\mathbb{P}^{1}$ into the manifold of two term partial
flags.

\medskip{}


\section{Monopoles and Instantons}

An $\text{SU}(N)$ instanton on $\mathbb{R}^{4}$ is a connection
1-form $A_{\square}$ on the (trivial) principal $\text{SU}(N)$ bundle
$P\rightarrow\mathbb{R}^{4}$ which satisfies the (anti-)self-duality
equations 
\vspace{0.3cm}
\[
F_{\square}=\pm\star F_{\square}
\]
where $F_{\square}$ is the curvature form of $A_{\square}$, and
the asymptotic decay condition, that $A_{\square}$ must extend to
a connection on $S^{4}$ (the conformal compactification of $\mathbb{R}^{4}$).
We will restrict to the anti-self-dual instantons. For an instanton,
the Yang-Mills Lagrangian 
\vspace{0.3cm}
\[
-\int_{\mathbb{R}^{4}}\text{Tr }F_{\square}\wedge\star F_{\square}
\]
is an $L^{2}$-norm of the curvature and is equal to $8\pi\kappa$
where $\kappa$ is an integer. $\kappa$ is a topological invariant
called the \textit{instanton charge}. (See \cite{key-6} for a complete
treatment.)

A magnetic monopole $(A,\phi)$ on $\mathbb{R}^{3}$ (euclidean) is
a connection 1-form $A$ on the principal $\text{SU}(N)$ bundle $P\rightarrow\mathbb{R}^{3}$
and a section $\phi$ of the adjoint bundle $\text{ad }P$ which satisfies
the Bogolmonyi equations 
\[
F_{A}=\star_{e}D_{A}\phi
\]
where the Hodge star dual $\star_{e}$ is defined by the euclidean
metric, and a choice of asymptotic decay conditions. The moduli of
euclidean monopoles is foliated by mass numbers $p_{1},\ldots,p_{N-1}\in\mathbb{R}$
and magnetic charge numbers $k_{1},\ldots,k_{N-1}\in\mathbb{Z}$.

A magnetic monopole $(A,\phi)$ in hyperbolic space $H^{3}$ can be
defined as an instanton on $\mathbb{R}^{4}$ invariant under the following
circle $S^{1}$ action \cite{key-4}. Choose coordinates $(x_{1},x_{2},x_{3},x_{4})$
for $\mathbb{R}^{4}$ and rotate the $x_{3}x_{4}$ plane with the
$x_{1}x_{2}$ plane as the axis of rotation. Then we may use new coordinates
$(x_{1},x_{2},r,\theta)$ where $e^{i\alpha}\in S^{1}$ acts by $\theta\mapsto\alpha\theta$.
The euclidean metric in these coordinates is 
\[
ds^{2}=r^{2}\left(\dfrac{dx_{1}^{2}+dx_{2}^{2}+dr^{2}}{r^{2}}+d\theta^{2}\right).
\]
Without the axis of rotation, $\mathbb{R}^{4}$ is foliated by upper
half spaces and this metric induces the Poincaré hyperbolic metric
on each. Conformally, 
\[
\mathbb{R}^{4}-\mathbb{R}^{2}\simeq S^{1}\times H^{3}.
\]
The instantons which are invariant under this circle action may be
interpreted as a connection $A$ on $H^{3}$ with all the right asymptotic
decay conditions following from the original instanton.

A monopole connection $A_{\square}$ in these coordinates is equivalent
to a potential $A=A_{x_{1}}dx_{1}+A_{x_{2}}dx_{2}+A_{r}dr$ and a
Higgs field $\phi$ (the $d\theta$ part), a section of the adjoint
bundle. The self-duality condition reduces to the \textit{hyperbolic}
Bogolmonyi equations 
\[
F_{A}=\star D_{A}\phi
\]
where the Hodge star $\star$ is defined by the above hyperbolic metric.

The moduli space of hyperbolic monopoles $(A,\phi)$ has components labelled by \textit{mass
numbers} $p_{1},\ldots,p_{N-1}\in\mathbb{Z}$ (or in $\frac{1}{2}+\mathbb{Z}$
if $N$ is even) which I order $p_1<p_2<\ldots<p_{N-1}$ and corresponding \textit{charge numbers} $k_{1},\ldots,k_{N-1}\in\mathbb{Z}$.
Note the restriction (compared to the euclidean case) on the mass
numbers which arise as the weights of the $S^{1}$-action; this is
a drawback of defining hyperbolic monopoles as $S^{1}$-invariant
instantons. For the rest of the paper, the mass numbers will be assumed
to be distinct; this is the case of maximal symmetry breaking where
the $\text{SU}(N)$ symmetry is reduced to the symmetry of a maximal
torus $\text{U}(1)^{N-1}$ which preserves $\phi$ at a point on the
conformal sphere at infinity.

To employ the ADHM construction \cite{key-6,key-7}, we need to work
in the twistor space $\mathbb{P}^{3}$ of $\mathbb{R}^{4}\subset S^{4}$.
Consider the fibration 
\begin{equation}
\mathbb{CP}^{3}\rightarrow\mathbb{HP}^{1}\simeq S^{4}\label{eq:twistor fibration}
\end{equation}
\[
[x:y:z:w]\mapsto[x+yj:z+wj].
\]
The left multiplication by $j\in\mathbb{H}$ leaves $S^{4}$ invariant
but induces an involution on $\mathbb{P}^{3}$ 
\[
J[x:y:z:w]=[\bar{y}:-\bar{x}:\bar{w}:-\bar{z}]
\]
acting as the antipodal map on the $\mathbb{P}^{1}$ fibres of the
twistor fibration\eqref{eq:twistor fibration}, commonly called a
``real structure'' on $\mathbb{P}^{3}$.

The Penrose-Ward transform is a correspondence between

\begin{minipage}[t]{1\columnwidth}%
\begin{enumerate}
\item instantons on $S^{4}$ realised as vector bundles with unitary structure
and a connection with anti-self-dual curvature, and 
\item holomorphic vector bundles $E$ on $\mathbb{P}^{3}$ with a real form.\end{enumerate}
\end{minipage}

The circle action on $\mathbb{R}^{4}$ lifts to $\mathbb{P}^{3}$
along this fibration as the action 
\[
[x:y:z:w]\mapsto[c^{-1/2}x:c^{1/2}y:c^{-1/2}z:c^{1/2}w]
\]
where $c\in S^{1}\subset\mathbb{C}^{\times}$.


\begin{figure}
\includegraphics[scale=0.1]{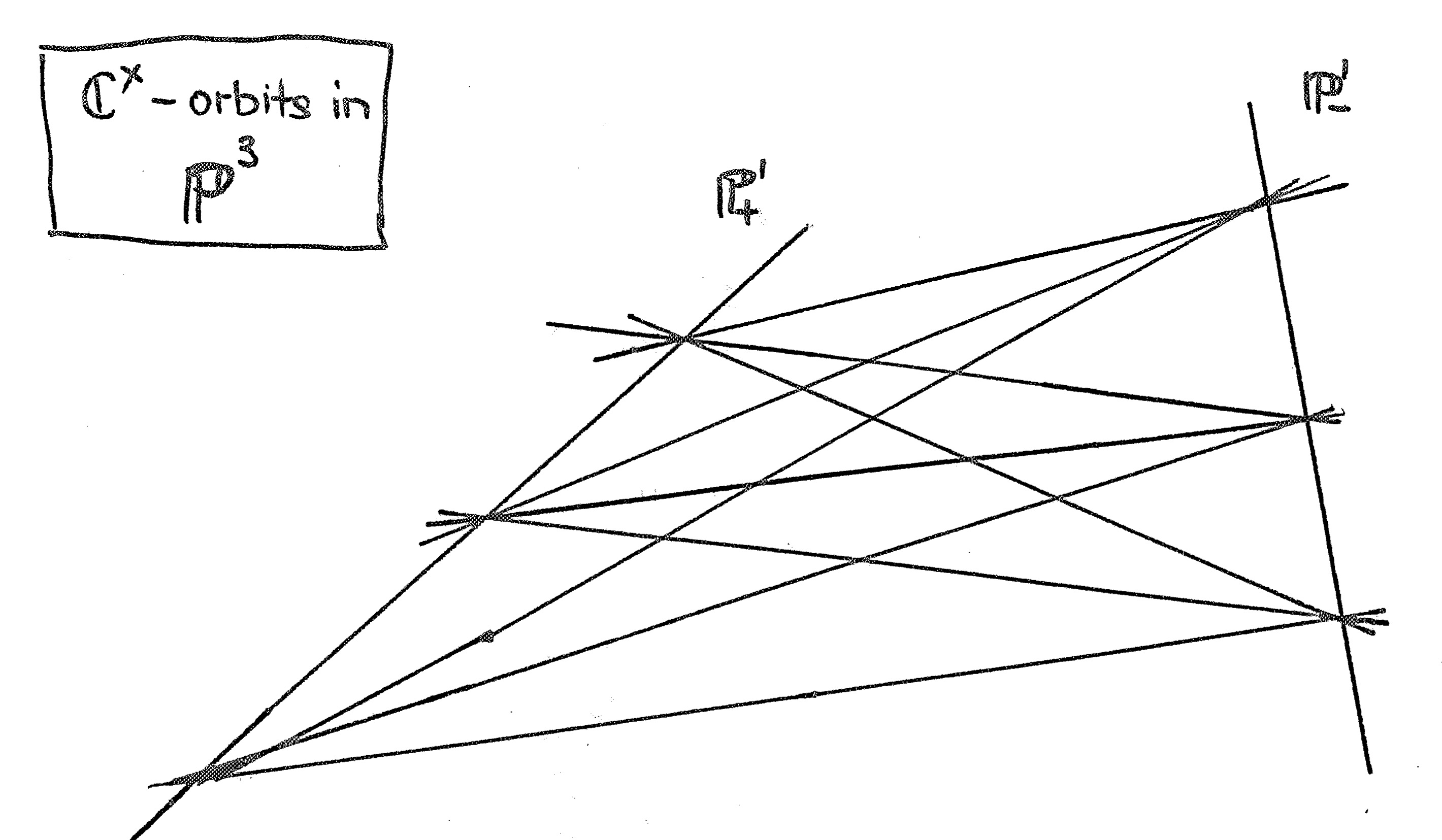}

\caption{The decomposition of $\mathbb{P}^{3}$ by the $\mathbb{C}^{\times}$-action
into fixed lines and $\mathbb{C}^{\times}$ orbits.}
\end{figure}


In $\mathbb{P}^{3}$, there are two fixed lines $\mathbb{P}_{+}^{1}=\left\{ \left[x:0:z:0\right]\right\} $
and $\mathbb{P}_{-}^{1}=\left\{ \left[0:y:0:w\right]\right\} $ of
the $\mathbb{C}^{\times}$-action which cover the fixed $S_{\partial H}^{2}\subset S^{4}$.
The $\mathbb{C}^{\times}$-action is free on $\mathbb{P}^{3}-\mathbb{P}_{+}^{1}\cup\mathbb{P}_{-}^{1}$
so we can decompose it into $\mathbb{C}^{\times}$-orbits. The boundary
of each $\mathbb{C}^{\times}$-orbit is a pair of points, one from
each fixed line and each point in $\mathbb{P}_{+}^{1}\times\mathbb{P}_{-}^{1}$
uniquely determines a $\mathbb{C}^{\times}$-orbit. Thus the space
of orbits 
\[
Q=\dfrac{\mathbb{P}^{3}-\mathbb{P}_{+}^{1}\cup\mathbb{P}_{-}^{1}}{\mathbb{C}^{\times}}
\]
is isomorphic to $\mathbb{P}^{1}\times\mathbb{P}^{1}$. $Q$ is known
as the hyperbolic monopole mini-twistor space.


\begin{figure}
\includegraphics[scale=0.1]{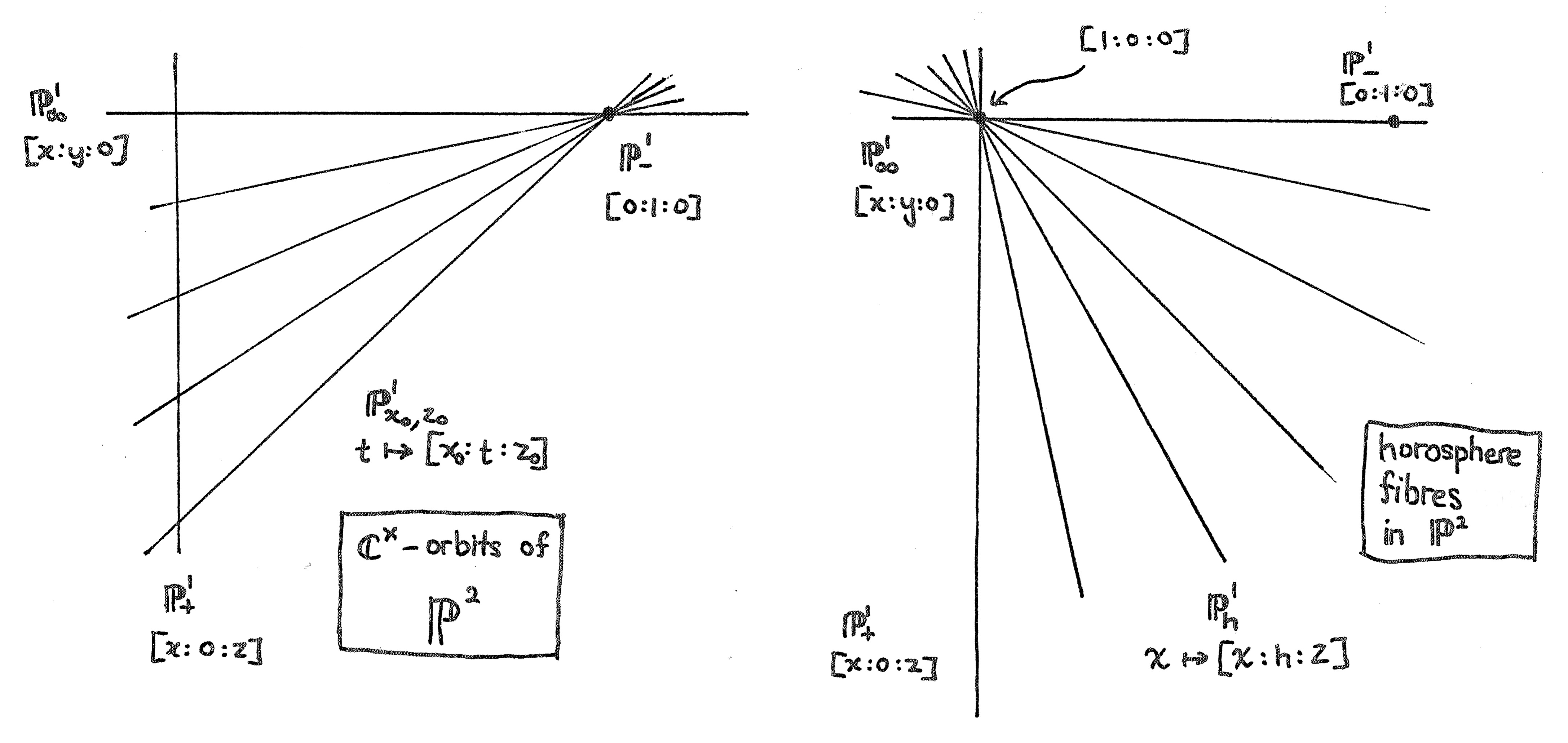}

\caption{The $\mathbb{C}^{\times}$ orbits of $\mathbb{P}^{2}$ and the fibres
of horospheres intersecting $\{\infty\}\in\partial H^{3}$.\label{fig:anatomy of P^2}}
\end{figure}


The projective plane $\mathbb{P}^{2}$ satisfying $w=0$ contains
the fixed line $\mathbb{P}_{+}^{1}$ and intersects $\mathbb{P}_{-}^{1}$
at a point $X_{-}$. This choice of $\mathbb{P}^{2}$ picks out a
unique point $\{\infty\}\in\partial H^{3}$ covered by $\mathbb{P}_{\infty}^{1}=\left\{ [x:y:0]\right\} $,
the only fibre over a point of $\partial H^{3}$ contained in $\mathbb{P}^{2}$.
Assume that $z=-1$ by projectivity and then $\mathbb{P}^{2}-\mathbb{P}_{+}^{1}$
is decomposed into a family of orbits $\{\mathbb{P}_{x_{0}}^{1}\}$
of the $\mathbb{C}^{\times}$-action, indexed by $x_{0}\in\mathbb{P}_{+}^{1}$
where the orbits intersect $\mathbb{P}_{+}^{1}$. $\mathbb{P}^{2}-\mathbb{P}_{+}^{1}$
also decomposes into a family of lines $\{\mathbb{P}_{t}^{1}\}_{t\in\mathbb{P}_{x_{0}}^{1}}$
(for some fixed choice of $x_{0}$) intersecting the point $\left[1:0:0\right]$
(the intersection of $\mathbb{P}_{+}^{1}$ and $\mathbb{P}_{\infty}^{1}$)
which map to horospheres in $H^{3}$ at $\{\infty\}$.

A framing of an instanton is an isomorphism $P_{\infty}\tilde{\rightarrow}\text{SU}(N)$
for the fibre of $P$ at the point at infinity of $S^{4}$. A framed
$\text{SU}(N)$ instanton is an instanton together with a framing.

The ADHM construction can be carried out over either $\mathbb{P}^{2}$
or $\mathbb{P}^{3}$. The $\mathbb{P}^{3}$ construction can always
yield the $\mathbb{P}^{2}$ construction via geometric invariant theory
but the converse is not true.

By a theorem of Donaldson \cite{key-8}, there is a natural correspondence
between framed instantons and holomorphic bundles on $\mathbb{P}^{2}\subset\mathbb{P}^{3}$
(with first Chern class $c_{1}=0$ since $\text{SU}(N)$ has determinant
1) with a fixed holomorphic trivialisation at the fibre $\mathbb{P}_{\infty}^{1}$
of infinity via the twistor fibration \eqref{eq:twistor fibration}.

Such a holomorphic bundle $E$ on $\mathbb{P}^{2}$ can be constructed
as the cohomology of monads \cite{key-9}. A monad over $\mathbb{P}^{2}$
is the following pair of maps 
\[
\begin{array}{ccccc}
\underline{H} & \overset{A_{X}}{\rightarrow} & \underline{K} & \overset{B_{X}}{\rightarrow} & \underline{L}\end{array}
\]
where 
\begin{enumerate}
\item $\underline{H}=H\otimes\mathcal{O}(-1)$, $\underline{K}=K\otimes\mathcal{O}$,
$\underline{L}=L\otimes\mathcal{O}(1)$; 
\item $H$,$K$,$L$ are $\kappa$,$\kappa+N$,$\kappa$ dimensional vector
spaces over $\mathbb{C}$ respectively; 
\item $\mathcal{O}(1)$ is the Hopf bundle over $\mathbb{P}^{2}$ and 
\item $A_{X}$,$B_{X}$ are linear maps for each $[x:y:z]=X\in\mathbb{P}^{2}$
and depend linearly on $X$. 
\end{enumerate}
The map $A_{X}$ needs to be injective, the map $B_{X}$ needs to
be surjective and $B_{X}A_{X}\equiv0_{\kappa}$. 

Since the maps $A_{X}$,$B_{X}$ vary holomorphically with $X\in\mathbb{P}^{2}$,
the holomorphic bundle $E$ can be defined fibre-wise by the cohomology
\[
E_{X}=\ker B_{X}/\text{im }A_{X}
\]
 of the monad. For an instanton, this construction is unique up to
an action of $\text{GL}_{HKL}=\text{GL}(H)\times\text{GL}(K)\times\text{GL}(L)$.

Following Donaldson, the conditions on $A_{X}$ and $B_{X}$ imply
that a basis can be chosen such that 
\[
A_{X}=\left[\begin{array}{c}
x+z\alpha_{1}\\
y+z\alpha_{2}\\
za
\end{array}\right]
\]
\[
B_{X}=\left[\begin{array}{ccc}
-y-z\alpha_{2} & x+z\alpha_{1} & zb\end{array}\right]
\]
where $\alpha_{1}$ and $\alpha_{2}$ are $\kappa\times\kappa$ matrices,
$a$ is a $N\times\kappa$ matrix, $b$ is a $\kappa\times N$ matrix
which we call ADHM matrices; they satisfy the complex ADHM equation
\begin{equation}
\left[\alpha_{1},\alpha_{2}\right]+ba=0.\label{eq:complex ADHM}
\end{equation}

The action of $\text{GL}_{HKL}$ on the monad induces the following
action of $\text{GL}(\kappa,\mathbb{C})\times\text{GL}(N,\mathbb{C})$
on the data $\alpha_{1},\alpha_{2},a$ and $b$ 
\[
\begin{aligned}\alpha_{i} & \mapsto g\alpha_{i}g^{-1}\\
a & \mapsto\lambda ag^{-1}\\
b & \mapsto gb\lambda^{-1}
\end{aligned}
\]
where $g\in\text{GL}(\kappa,\mathbb{C})$ and $\lambda\in\text{GL}(N,\mathbb{C})$.
We call this a ``gauge transformation'' of the ADHM data.

For the fibre $\mathbb{P}_{\infty}^{1}=\{[x:y:0]\}$ over infinity,
\[
A_{X}=\left[\begin{array}{c}
xI_{\kappa}\\
yI_{\kappa}\\
0_{N\times\kappa}
\end{array}\right]
\]
\[
B_{X}=\left[\begin{array}{ccc}
-yI_{\kappa} & xI_{\kappa} & 0_{\kappa\times N}\end{array}\right].
\]
Thus the trivialisation $\Psi:E\vert_{\mathbb{P}_{\infty}^{1}}\rightarrow\mathbb{C}^{N}$
fixes a basis (the ``frame'') for the last $N$ entries of $K$. 

The ADHM construction over $\mathbb{P}^{3}$ can be expressed in the
same way but with a dependence on the coordinate $w$ and an isomorphism
$\overline{J^{\ast}(E)}\cong E^{\ast}$ that covers the real structure
$J$ on $\mathbb{P}^{3}$ (See \cite{key-6,key-8} for details).

The maps $A_{X}$ and $B_{X}$ over $\mathbb{P}^{3}$ are 
\[
A_{X}=\left[\begin{array}{c}
x+z\alpha_{1}-w\alpha_{2}^{\ast}\\
y+z\alpha_{2}+w\alpha_{1}^{\ast}\\
za+wb^{\ast}
\end{array}\right]
\]
\[
B_{X}=\left[\begin{array}{ccc}
-y-z\alpha_{2}-w\alpha_{1}^{\ast} & x+z\alpha_{1}-w\alpha_{2}^{\ast} & zb-wa^{\ast}\end{array}\right].
\]
They satisfy both the complex ADHM equation \eqref{eq:complex ADHM}
and the real ADHM equation 
\begin{equation}
\mu=\left[\alpha_{1},\alpha_{1}^{\ast}\right]+\left[\alpha_{2},\alpha_{2}^{\ast}\right]+bb^{\ast}-a^{\ast}a=0
\end{equation}
which is a moment map $\mu:\mathbb{C}^{2(\kappa+N)}\rightarrow\mathfrak{u}(\kappa)$
for the system. This equation is only preserved by the subgroup of
$\text{GL}(\kappa,\mathbb{C})$ whose elements obey $g^{-1}=g^{\ast}$.
Thus there is a reduction to an action of $\text{U}(\kappa)\times U(N)$
on the data $\alpha_{1},\alpha_{2},a$ and $b$.

The holomorphic vector bundle constructed on $\mathbb{P}^{3}$ agrees
with the bundle constructed over $\mathbb{P}^{2}$ for the same ADHM
data $(\alpha_{1},\alpha_{2},a,b)$ - we will call them both $E$.

Over the fixed line $\mathbb{P}_{+}^{1}$, the $\mathbb{C}^{\times}$-action
induces a representation on the fibres of the holomorphic vector bundle
$E$. All the irreducible representations of $\mathbb{C}^{\times}$
are 1-dimensional so up to conjugation, the circle action (for $\text{SU}(N)$)
takes the form 
\[
c\mapsto\lambda(c)=\left[\begin{array}{cccc}
c^{p_{1}}\\
 & \ddots\\
 &  & c^{p_{N-1}}\\
 &  &  & c^{p_{N}}
\end{array}\right]
\]
where $p_{1}<\ldots<p_{N-1}$ (since they are assumed to be distinct)
are the weights of the $\mathbb{C}^{\times}$-action and they are
either all integers or if $N$ is even, they can also be all half-integers
(this is well-defined since the $\mathbb{C}^{\times}$ action comes
from a double cover of $\mathbb{C}^{\times}$). Since the structure
group is $\text{SU}(N)$, $p_{N}=-p_{1}-\ldots-p_{N-1}$.

To study hyperbolic monopoles via the ADHM construction, we examine
what it means for a monad to be ``circle invariant''. Work has been
done in this direction by Norbury in his PhD thesis \cite{key-10}
for the $\text{SU}(2)$ case; however, his results apply equally to
the $\text{SU}(N)$ case. Since this PhD thesis is not widely available,
a proof will be supplied. 
\begin{prop}[Norbury]
\label{prop: Cstar monad} A monad over $\mathbb{P}^{2}$ whose cohomology
is a holomorphic $\mathbb{C}^{N}$-vector bundle with trivialisation
data corresponding to a framed instanton on $\mathbb{R}^{4}$ is $\mathbb{C}^{\times}$-invariant
if and only if there exists a homomorphism $P_{c}:\mathbb{C}^{\times}\rightarrow\text{GL}(\kappa,\mathbb{C})$
such that

\begin{minipage}[t]{1\columnwidth}%
\begin{enumerate}
\item $\alpha_{1}=P_{c}\alpha_{1}P_{c}^{-1}$ 
\item $\alpha_{2}=cP_{c}\alpha_{2}P_{c}^{-1}$ 
\item $a=\lambda aP_{c}^{-1}$ 
\item $b=cP_{c}b\lambda^{-1}$\end{enumerate}
\end{minipage}\end{prop}
\begin{proof}
For the monopole to be $\mathbb{C}^{\times}$-invariant, the monad
maps need to be $\mathbb{C}^{\times}$-equivariant. There needs to
be an element $\left(\sigma,\rho,\sigma'\right)$ of $\text{GL}_{HKL}$
for which the maps $A_{X}$ and $B_{X}$ satisfy $\rho(c)A_{(x,y,z)}=A_{(x,cy,z)}\sigma(c)$
and $\sigma'(c)B_{(x,y,z)}=B_{(x,cy,z)}\rho(c)$. We can ask that
the choice of basis made for $K$ be preserved which means that $\rho(c)$
should split into blocks on the diagonal, $\text{diag}\left(\rho_{1},\rho_{2},\rho_{3}\right)\in\text{GL}(\kappa,\mathbb{C})\times\text{GL}(\kappa,\mathbb{C})\times\text{GL}(N,\mathbb{C})$.

The condition $A_{(x,cy,z)}=\rho(c)A_{(x,y,z)}\sigma^{-1}(c)$ in
this basis is 
\[
\left[\begin{array}{c}
x+z\alpha_{1}\\
y+z\alpha_{2}\\
za
\end{array}\right]\mapsto\left[\begin{array}{c}
x+z\alpha_{1}\\
cy+z\alpha_{2}\\
za
\end{array}\right]=\text{diag}\left(\rho_{1},\rho_{2},\rho_{3}\right)\left[\begin{array}{c}
x+z\alpha_{1}\\
y+z\alpha_{2}\\
za
\end{array}\right]\sigma^{-1}.
\]
Note that $x=\rho_{1}x\sigma^{-1}$ implies that $\rho_{1}=\sigma$
and $cy=\rho_{2}y\sigma^{-1}$ implies that $\rho_{2}=c\sigma$.

Likewise, $B_{(x,cy,z)}=\sigma'(c)B_{(x,y,z)}\rho^{-1}(c)$ in the
chosen basis reads as 
\[
\left[\begin{array}{ccc}
-cy-z\alpha_{2} & x+z\alpha_{1} & zb\end{array}\right]=\sigma'\left[\begin{array}{ccc}
-y-z\alpha_{2} & x+z\alpha_{1} & zb\end{array}\right]\text{diag}\left(\rho_{1}^{-1},\rho_{2}^{-1},\rho_{3}^{-1}\right).
\]
From the first two blocks, $-cy=-\sigma'y\rho_{1}^{-1}$ implies that
$c\rho_{1}=\sigma'$ and $x=\sigma'x\rho_{2}^{-1}$ implies that $\rho_{2}=\sigma'$.

Together, this means $\sigma=P_{c}=\rho_{1}$ and $\sigma'=cP_{c}=\rho_{2}$
for some $P_{c}\in\text{GL}(\kappa,\mathbb{C})$. Recall that the
last $N$ basis elements of $K$ provide the framing so $\rho_{3}$
needs to be the representation $\lambda_{c}$. Thus, the conditions
(1)-(4) of the theorem are exactly the conditions for the $\mathbb{C}^{\times}$-equivariance
of $A_{X}$ and $B_{X}$. 
\end{proof}
Thus we see that in the case of a circle invariant monopole, the $\mathbb{C}^{\times}$-action
on the monad's bundles is multiplication by 
\[
c\mapsto\text{diag}\left(P_{c},\text{diag}\left(P_{c},cP_{c},\lambda_{c}\right),cP_{c}\right)\in\text{GL}(H)\times\text{GL}(K)\times\text{GL}(L).
\]
The homomorphism $P_{c}$ is a representation of $\mathbb{C}^{\times}$
so we can diagonalise it. This means that $H$, $K$ and $L$ can
be decomposed into weight spaces for the $\mathbb{C}^{\times}$-action.
The ADHM data $\alpha_{1},\alpha_{2},a,b$ must then preserve these
weight spaces.

Austin and Braam \cite{key-3} found the weight space decomposition
for the $\text{SU}(2)$ case via the equivariant index theorem. In
the next section, we will see a calculation of the weight spaces for
any $\text{SU}(N)$. It is enough to compute the $\mathbb{C}^{\times}$-representation
$P_{c}$ over the fixed line $\mathbb{P}_{+}^{1}$ since this is enough
to determine the ADHM data $(\alpha_{1},\alpha_{2},a,b)$.


\section{A Chern Characters Calculation}

The starting point of the calculation is the following display (which
can be found in \cite{key-9}) for a monad

\begin{equation}
\xymatrix{ &  & 0\ar[d] & 0\ar[d]\\
0\ar[r] & \underline{H}\ar[r]\ar@{=}[d] & \ker B_{X}\ar[r]\ar[d] & E\ar[r]\ar[d] & 0\\
0\ar[r] & \underline{H}\ar[r]^{A_{X}} & \underline{K}\ar[r]\ar[d]_{B_{X}} & \text{coker}\, A_{X}\ar[r]\ar[d] & 0\\
 &  & \underline{L}\ar[d]\ar@{=}[r] & \underline{L}\ar[d]\\
 &  & 0 & 0
}
\label{eq:monad display}
\end{equation}
where the rows and columns are all exact.

The equivariant Chern character of $\mathbb{P}^{1}$ is a map $K_{\mathbb{C}^{\times}}(\mathbb{P}^{1})\rightarrow H_{\mathbb{C}^{\times}}^{\ast}(\mathbb{P}^{1})$,
from the equivariant K-theory to the equivariant cohomology of a space
$\mathbb{P}^{1}$. By the additivity of the Chern character, the right
vertical and bottom horizontal exact sequences of the display gives
us the following 
\[
\text{ch}(\text{coker }A_{X})=\text{ch}(E)+\text{ch}(\underline{L})
\]
\[
\text{ch}(\underline{K})=\text{ch}(\underline{H})+\text{ch}(\text{coker }A_{X})
\]
where $\text{ch}$ denotes the \textit{$\mathbb{C}^{\times}$-equivariant}
Chern character. Putting them together yields 
\begin{equation}
\text{ch}(E)=\text{ch}(\underline{K})-\text{ch}(\underline{H})-\text{ch}(\underline{L}).\label{eq: ch display}
\end{equation}
The upshot is that if we know the equivariant Chern character of the
holomorphic bundle $E$, we can compute the equivariant Chern character
of the monad vector spaces $H$,$K$ and $L$ over $\mathbb{P}_{+}^{1}$
and hence their $\mathbb{C}^{\times}$ weight decomposition. Concretely,
this data is encoded in the exponents of the matrix $P_{c}$ and will
induce a decomposition of the ADHM matrices.

Since the bundle $E$ is trivial over $\mathbb{P}_{+}^{1}$, we have
a representation of $\mathbb{C}^{\times}$ on the fibres which allows
us to compute the equivariant Chern character of $E\vert_{\mathbb{P}_{+}^{1}}$.
Over any $\mathbb{P}^{1}$, all holomorphic vector bundles split into
line bundles by the Birkoff-Grothendieck splitting principle \cite{key-9}.
The strategy is to localise to $\mathbb{P}_{+}^{1}$, split all the
relevant bundles and compute the exponents of $P_{c}$. Since the
ADHM matrices are constant, any conditions on them over any line will
hold globally.


\subsection{The bundle E}

~

For $\text{SU}(2)$, Atiyah showed that over $\mathbb{P}_{+}^{1}$,
$E=\mathcal{O}(k)\otimes\mathcal{L}^{-p}\oplus\mathcal{O}(-k)\otimes\mathcal{L}^{p}$
where $\mathcal{L}^{p}$ is the trivial line bundle with the $c^{p}$
representation of $\mathbb{C}^{\times}$ \cite{key-4}. This follows
from a result of equivariant K-theory that over a fixed point set
$M$, 
\[
K_{\mathbb{C}^{\times}}(M)=K(M)\otimes R(\mathbb{C}^{\times})
\]
where $R(\mathbb{C}^{\times})=\mathbb{Z}[u]$ is the ring of characters
of the representations of $\mathbb{C}^{\times}$ \cite{key-11}.

The $\mathbb{C}^{\times}$-representation on $E$ over $\mathbb{P}_{+}^{1}$
\[
c\mapsto\lambda(c)=\text{diag}\left(\begin{array}{ccc}
c^{p_{1}} & \ldots & c^{p_{N}}\end{array}\right)
\]
ordered $p_1<p_2<\ldots<p_N$ splits $E$ into a sum of line bundles. Since these line bundles are
algebraic, we invoke Birkhoff-Grothendieck {[}Okonek-Schneider-Spindler
1980{]} to see the unique splitting 
\[
E=\mathcal{O}(k_{1})\otimes\mathcal{L}^{p_{1}}\oplus\ldots\oplus\mathcal{O}(k_{N-1})\otimes\mathcal{L}^{p_{N-1}}\oplus\mathcal{O}\left(k_{N}\right)\mathcal{L}^{p_{N}}
\]
where $k_{N}=-(k_{1}+\ldots+k_{N-1})$ and $p_{N}=-(p_{1}+\ldots+p_{N-1})$.

Using results in \cite{key-4,key-12}, we calculate the equivariant
first Chern class and the total Chern class of $E$. The equivariant
first Chern class of a line bundle of the form $\mathcal{O}(k)\otimes\mathcal{L}^{p}$
is 
\[
c_{1}^{eq}=kx+pu
\]
where $x$ is the second degree generator of the usual $H^{2}(\mathbb{P}^{1})$
and $u$ is the first degree generator of $R(\mathbb{C}^{\times})$.

This is enough to calculate the equivariant Chern character 
\[
\text{ch}(E)=e^{k_{1}x+p_{1}u}+\ldots+e^{k_{N}x+p_{N}u}
\]
and since $H^{\ast}(\mathbb{P}^{1})=\mathbb{Z}[x]/\left\langle x^{2}\right\rangle $,
the following series expansion with respect to $x$ is exact 
\begin{equation}
\begin{split}\text{ch}(E)=\  & e^{p_{1}u}+\ldots+e^{p_{N}u}\\
 & +x\left(k_{1}e^{p_{1}u}+\ldots+k_{N}e^{p_{N}u}\right).
\end{split}
\end{equation}

The equivariant total Chern class of $E$ is given by 
\[
\prod_{i=1}^{N}(1+k_{i}x+p_{i}u)\quad\mod x^{2}.
\]
The localisation formula from Atiyah and Bott \cite{key-12} tells
us that the second Chern class $c_{2}$ (remember that $c_{1}(E)=0$)
can be found by looking at the coefficient of
$x$ and dividing it by $u$. This is \emph{positive} integer
\begin{equation}
c_{2}(E)=-\left[2\sum_{i=1}^{N-1}k_{i}p_{i}+\underset{{\scriptstyle i<j}}{\sum_{i=1}^{N-2}}\left(k_{i}p_{j}+k_{j}p_{i}\right)\label{eq:c2E}\right]
\end{equation}
which reduces to $2kp$ as expected for the $\text{SU}(2)$ case $p_1=-p$ which
is known.


\subsection{The main calculation} \hfill

Since the $x$-terms in the Chern character of $E$ only has terms
up to $e^{p_{1}u}$ and $e^{p_{N}u}$, the lowest weight of $P_{c}$
and highest weight of $cP_{c}$ are $c^{p_{1}}$ and $c^{p_{N}}$
respectively. This is required because for the $x$-terms, the lowest
weight term of $\underline{H}$ and the highest weight term of $\underline{L}$
do not cancel with any other terms on the right side of (\ref{eq: ch display})
and therefore must exactly match $x$-terms of $\text{ch}(E)$.

The homomorphism $P_{c}$ has the form

\[
\begin{array}{cc}
\text{diag} & \left(\begin{array}{cccccccccc}
c^{p_{1}} & \ldots & c^{p_{1}} & c^{p_{1}+1} & \ldots & c^{p_{1}+1} & \ldots & c^{p_{N}-1} & \ldots & c^{p_{N}-1}\end{array}\right)\\
 & \begin{array}{cccccccccc}
\longleftarrow & \chi_{p_{1}} & \longrightarrow & \longleftarrow & \chi_{p_{1}+1} & \longrightarrow & \ldots & \longleftarrow\  & \chi_{p_{N}-1} & \longrightarrow\end{array}
\end{array}
\]
and the $p_{N}-p_{1}$ numbers $\chi_{p_{1}},\ldots,\chi_{p_{N}-1}$
are what we need to calculate.

The vector bundles $\underline{H}$,$\underline{K}$ and $\underline{L}$
decompose as follows: 
\[
\underline{H}=\bigoplus_{i=p_{1}}^{p_{N}-1}\left(\mathcal{O}(-1)\otimes\mathcal{L}^{i}\right)^{\oplus\chi_{i}}
\]

\[
\underline{K}=\bigoplus_{i=p_{1}}^{p_{N}-1}\left(\mathcal{L}^{i}\right)^{\oplus\chi_{i}}\oplus\bigoplus_{i=p_{1}}^{p_{N}-1}\left(\mathcal{L}^{i+1}\right)^{\oplus\chi_{i+1}}\oplus\left(\mathcal{L}^{p_{1}}\oplus\ldots\oplus\mathcal{L}^{p_{N}}\right)
\]

\[
\underline{L}=\bigoplus_{i=p_{1}}^{p_{N}-1}\left(\mathcal{O}(1)\otimes\mathcal{L}^{i+1}\right)^{\oplus\chi_{i+1}}.
\]
Note that $\underline{K}$ has been arranged into the parts on which
the $\mathbb{C}^{\times}$-action is via $P_{c}$, $cP_{c}$ and $\lambda$
respectively.

\vspace{2cm}
\pagebreak

The corresponding equivariant Chern characters are: 
\[
\begin{split}\text{ch}(\underline{H}) & =\sum_{i=p_{1}}^{p_{N}-1}\chi_{i}e^{-x+iu}\\
 & =\sum_{i=p_{1}}^{p_{N}-1}\chi_{i}e^{iu}-x\left(\sum_{i=p_{1}}^{p_{N}-1}\chi_{i}e^{iu}\right)
\end{split}
\]

\begin{equation}
\begin{split}\text{ch}(\underline{K}) & =\sum_{i=p_{1}}^{p_{N}-1}\chi_{i}e^{iu}+\sum_{i=p_{1}}^{p_{N}-1}\chi_{i}e^{(i+1)u}+\left(e^{p_{1}u}+\ldots+e^{p_{N}u}\right)\\
 & =\chi_{p_{1}}e^{p_{1}u}+\sum_{i=p_{1}+1}^{p_{N}-1}(\chi_{i-1}+\chi_{i})e^{iu}+\chi_{p_{N-1}}e^{p_{N}u}+\left(e^{p_{1}u}+\ldots+e^{p_{N}u}\right)
\end{split}
\label{eq:chK}
\end{equation}

\[
\begin{split}\text{ch}(\underline{L}) & =\sum_{i=p_{1}}^{p_{N}-1}\chi_{i}e^{x+(i+1)u}\\
 & =\sum_{i=p_{1}}^{p_{N}-1}\chi_{i}e^{(i+1)u}+x\left(\sum_{i=p_{1}}^{p_{N}-1}\chi_{i}e^{(i+1)u}\right).
\end{split}
\]

We proceed by comparing coefficients. The $x$-terms are enough to
determine the unknowns $\chi_{p_{1}},\ldots,\chi_{p_{N}-1}$. 
\[
xe^{p_{1}u}:k_{1}=\chi_{p_{1}}
\]
\[
xe^{p_{N}u}:k_{N}=-\chi_{p_{N}}
\]
\[
xe^{p_{i}u},\text{ for }1<i\leq N-1:k_{i}=\chi_{p_{i}}-\chi_{p_{i}-1}
\]
and all the other $x$-terms require that $\chi_{j}=\chi_{j-1}$ when
$j\neq p_{i}$ for any of the $1\leq i\leq N$.

The interesting 1-terms are the ones of the form $e^{p_{i}u}$. The
rightmost terms of (\ref{eq:chK}) supply the 1-terms of $\text{ch}(E)$.
We expected to see this because in the monad, $\underline{K}$ carries
the trivialisation/framing data of $E$ in its last $N$ basis elements.
The rest of the 1-terms $\text{ch}(\underline{K})$ cancel with the
1-terms of $\text{ch}(\underline{H})$ and $\text{ch}(\underline{L})$
to show that they are consistent with the constraints set by the $x$-terms.

In the case of $\text{SU}(3)$, the weights run from $p_{1}$ to
$p_{2}$ with coefficients $\chi_{i}=k_{1}$ and then from $p_{2}$
to $-p_{1}-p_{2}$ with coefficients $\chi_{i}=k_{1}+k_{2}$ . At $p_{2}$,
the coefficient jumps from $\chi_{p_{2}-1}=k_{1}$ to $\chi_{p_{2}}=k_{1}+k_{2}$.
This is illustrated by the following diagram (which should be viewed
as an interval - the domain of an evolution equation) 
\[
\xymatrix{\underset{p_{1}}{\centerdot}\ar@{-}[rr]_{k_{1}}^{p_{2}-p_{1}} &  & \underset{p_{2}}{\centerdot}\ar@{-}[rr]_{k_{1}+k_{2}}^{-2p_{2}-p_{1}} &  & \underset{p_{3}=-p_{1}-p_{2}}{\centerdot}}
\]
where the quantity above the line is the number of distinct weights
with corresponding coefficient being the quantity under the line.
The dimensions of $P_{c}$ (as a square matrix) are given by 
\[
\left(p_{2}-p_{1}\right)k_{1}-\left(2p_{2}+p_{1}\right)\left(k_{1}+k_{2}\right)=-(2p_{1}k_{1}+2p_{2}k_{2}+p_{1}k_{2}+p_{2}k_{1})
\]
which is exactly the formula for the second Chern class $c_{2}(E)$
from the previous subsection.

In general, we have 
\[
\xymatrix{\underset{p_{1}}{\centerdot}\ar@{-}[rr]_{k_{1}}^{p_{2}-p_{1}} &  & \underset{p_{2}}{\centerdot} & \cdots & \underset{p_{N-2}}{\centerdot}\ar@{-}[rr]_{k_{1}+\ldots+k_{N-2}}^{p_{N-1}-p_{N-2}} &  & \underset{p_{N-1}}{\centerdot}\ar@{-}[rr]_{k_{1}+\ldots+k_{N-1}}^{p_{N}-p_{N-1}} &  & \underset{p_{N}}{\centerdot}}
\]
and this gives us the dimensions of $P_{c}$ 
\begin{equation}
\kappa=\sum_{i=1}^{N-1}\left[\left(p_{i+1}-p_{i}\right)\sum_{j=1}^{i}k_{j}\right].\label{eq:kappa}
\end{equation}

In \cite{key-10}, Norbury proved the $\text{SU}(2)$ case of the
following proposition by a different method.
\begin{prop}
The dimensions $\kappa\times\kappa$ of $P_{c}$ are given by $\kappa=c_{2}(E)$
for all $G=\text{SU}(N)$, $N\in\mathbb{N}_{\geq3}$.\end{prop}
\begin{proof}
We proceed by induction. The $\text{SU}(3)$ case above is our base
step. (For the $\text{SU}(2)$ case, it is compatible too; $c_{2}(E)=2kp=\kappa$.)

For the inductive step, we assume that the proposition holds for $\text{SU}(N-1)$.
The difference in (\ref{eq:c2E}) between the $N$ and $N-1$ cases
is 
\begin{multline*}
(p_{N-2}-p_{N-1})(k_{1}+\ldots+k_{N-2})-(2p_{N-1}+p_{N-2}+\ldots+p_{1})(k_{1}+\ldots+k_{N-1})\\
+(2p_{N-2}+p_{N-3}+\ldots+p_{1})(k_{1}+\ldots+k_{N-2})\\
=-p_{N-1}(k_{1}+\ldots+k_{N-2})-(2p_{N-1}+p_{N-2}+\ldots+p_{1})k_{N-1}
\end{multline*}
which is exactly the extra terms of $c_{2}(E)$ in (\ref{eq:kappa})
in going from $N-1$ to $N$. 
\end{proof}


\subsection{Discrete Nahm equations}

~

The preceding section proves that 
\begin{prop}
Let $E$ be a $\mathbb{C}^{\times}$-equivariant holomorphic vector
bundle on $\mathbb{P}^{3}$ ($\mathbb{C}^{\times}$-action $[x:y:z:w]\mapsto[c^{-1/2}x:c^{1/2}y:c^{-1/2}z:c^{1/2}w]$)
corresponding to a monopole with mass numbers $p_{1},\ldots,p_{N-1}\in\mathbb{Z}$
(or $\frac{1}{2}+\mathbb{Z}$ if $N$ is even) ordered $p_{1}<\ldots<p_{N-1}$,
and charge numbers $k_{1},\ldots,k_{N-1}\in\mathbb{Z}$. 

Then the
$\mathbb{C}^{\times}$ weight space decomposition of the monad
\[
\underline{H}\overset{A_{X}}{\rightarrow}\underline{K}\overset{B_{X}}{\rightarrow}\underline{L}
\]
restricted to $\mathbb{P}_{+}^{1}$ is 
\[
\underline{H}=\mathbb{C}_{p_{1}}^{k_{1}}\oplus\ldots\oplus\mathbb{C}_{p_{2}-1}^{k_{1}}\oplus\mathbb{C}_{p_{2}}^{k_{1}+k_{2}}\oplus\mathbb{C}_{p_{2}+1}^{k_{1}+k_{2}}\oplus\ldots\oplus\mathbb{C}_{p_{N}-1}^{-k_{N}}
\]
\[
\underline{K}=\mathbb{C}_{p_{1}}^{k_{1}+1}\oplus\mathbb{C}_{p_{1}+1}^{2k_{1}}\oplus\ldots\oplus\mathbb{C}_{p_{2}-1}^{2k_{1}}\oplus\mathbb{C}_{p_{2}}^{2\left(k_{1}+k_{2}\right)+1}\oplus\mathbb{C}_{p_{2}+1}^{2\left(k_{1}+k_{2}\right)}\oplus\ldots\oplus\mathbb{C}_{p_{N}-1}^{2\left(k_{1}+\ldots+k_{N-1}\right)}\oplus\mathbb{C}_{p_{N}}^{-k_{N}+1}
\]
\[
\underline{L}=\mathbb{C}_{p_{1}+1}^{k_{1}}\oplus\ldots\oplus\mathbb{C}_{p_{2}}^{k_{1}}\oplus\mathbb{C}_{p_{2}+1}^{k_{1}+k_{2}}\oplus\mathbb{C}_{p_{2}+2}^{k_{1}+k_{2}}\oplus\ldots\oplus\mathbb{C}_{p_{N}}^{-k_{N}}
\]
where the subscript denotes the weight of the $\mathbb{C}^{\times}$
representation on that component. The final mass and charge numbers
are defined $p_{N}=-\sum_{i=1}^{N-1}p_{i}$ and $k_{N}=-\sum_{i=1}^{N-1}k_{i}$
respectively. 
\end{prop}
Note that anti-self-dual instantons have instanton charge $\kappa<0$
which constrains the allowed mass and charge numbers of a hyperbolic
monopole.

The conditions of Proposition \ref{prop: Cstar monad} imply that
the ADHM data $(\alpha_{1},\alpha_{2},a,b)$ for a magnetic monopole
only map between components of the same weight. Now I will describe
the form of the ADHM data $(\alpha_{1},\alpha_{2},a,b)$ which preserve
the above weight decomposition.


\begin{figure}[h]
\begin{tikzpicture}[>= latex,descr/.style={fill=white,inner sep=2.5pt}]  \matrix (m) [matrix of math nodes,row sep=.07em,column sep=12em,minimum width=2em] {
&\mathbb{C}_{-3}^{k_{1}}\\ 
&\mathbb{C}_{-2}^{k_{1}}\\ 
&\mathbb{C}_{-1}^{k_{1}+k_{2}}\\ 
&\mathbb{C}_{0}^{k_{1}+k_{2}}\\ 
&\mathbb{C}_{1}^{k_{1}+k_{2}}\\ 
&\mathbb{C}_{2}^{k_{1}+k_{2}}\\ 
\mathbb{C}_{-3}^{k_{1}}&\mathbb{C}_{3}^{k_{1}+k_{2}}&\mathbb{C}_{-2}^{k_{1}}\\ \mathbb{C}_{-2}^{k_{1}}&\mathbb{C}_{-2}^{k_{1}}&\mathbb{C}_{-1}^{k_{1}}\\ \mathbb{C}_{-1}^{k_{1}+k_{2}}&\mathbb{C}_{-1}^{k_{1}}&\mathbb{C}_{0}^{k_{1}+k_{2}}\\ \mathbb{C}_{0}^{k_{1}+k_{2}}&\mathbb{C}_{0}^{k_{1}+k_{2}}&\mathbb{C}_{1}^{k_{1}+k_{2}}\\ \mathbb{C}_{1}^{k_{1}+k_{2}}&\mathbb{C}_{1}^{k_{1}+k_{2}}&\mathbb{C}_{2}^{k_{1}+k_{2}}\\ \mathbb{C}_{2}^{k_{1}+k_{2}}&\mathbb{C}_{2}^{k_{1}+k_{2}}&\mathbb{C}_{3}^{k_{1}+k_{2}}\\ \mathbb{C}_{3}^{k_{1}+k_{2}}&\mathbb{C}_{3}^{k_{1}+k_{2}}&\mathbb{C}_{4}^{k_{1}+k_{2}}\\ 
&\mathbb{C}_{4}^{k_{1}+k_{2}}\\ &\mathbb{C}_{-3}\\ &\mathbb{C}_{-1}\\ &\mathbb{C}_{4}\\ 
}; 
\path[-stealth] 
(m-7-1.east) edge node[descr,sloped]{$\alpha_1$} (m-1-2.mid west) edge node[descr,sloped,right=6pt]{$a_{-3}$} (m-15-2.west) 
(m-8-1.east) edge (m-2-2.mid west) edge (m-8-2)  
(m-9-1.east) edge (m-3-2.mid west) edge (m-9-2) edge node[descr,sloped]{$a_{-1}$} (m-16-2.west) 
(m-10-1.east) edge (m-4-2.mid west) edge (m-10-2) 
(m-11-1.east) edge (m-5-2.mid west) edge (m-11-2) 
(m-12-1.east) edge (m-6-2.mid west) edge (m-12-2) (m-13-1.east) edge (m-7-2.mid west) edge node[descr,sloped]{$\alpha_2$} (m-13-2)  
(m-2-2.mid east) edge node[descr,sloped]{$\alpha_2$} (m-7-3.west) 
(m-3-2.mid east) edge (m-8-3.west) 
(m-4-2.mid east) edge (m-9-3.west) 
(m-5-2.mid east) edge (m-10-3.west) 
(m-6-2.mid east) edge (m-11-3.west) 
(m-7-2.mid east) edge (m-12-3.west) 
(m-8-2.mid east) edge (m-7-3.base west) 
(m-16-2.east) edge node[descr,sloped]{$b_{-1}$} (m-8-3.base west) 
(m-9-2) edge (m-8-3.185) 
(m-10-2) edge (m-9-3.base west) 
(m-11-2) edge (m-10-3.base west) 
(m-12-2) edge (m-11-3.base west) 
(m-13-2) edge (m-12-3.base west) (m-14-2) edge node[descr,sloped]{$\alpha_1$} (m-13-3) 
(m-17-2.east) edge node[descr,sloped]{$b_{4}$} (m-13-3);
\end{tikzpicture} 

\caption{The weight decomposition of the monad of an $\text{SU}(3)$ hyperbolic
monopole with $p_{1}=-3$ and $p_{2}=-1$ (hence $\kappa=7k_{1}+5k_{2}$).}
\end{figure}
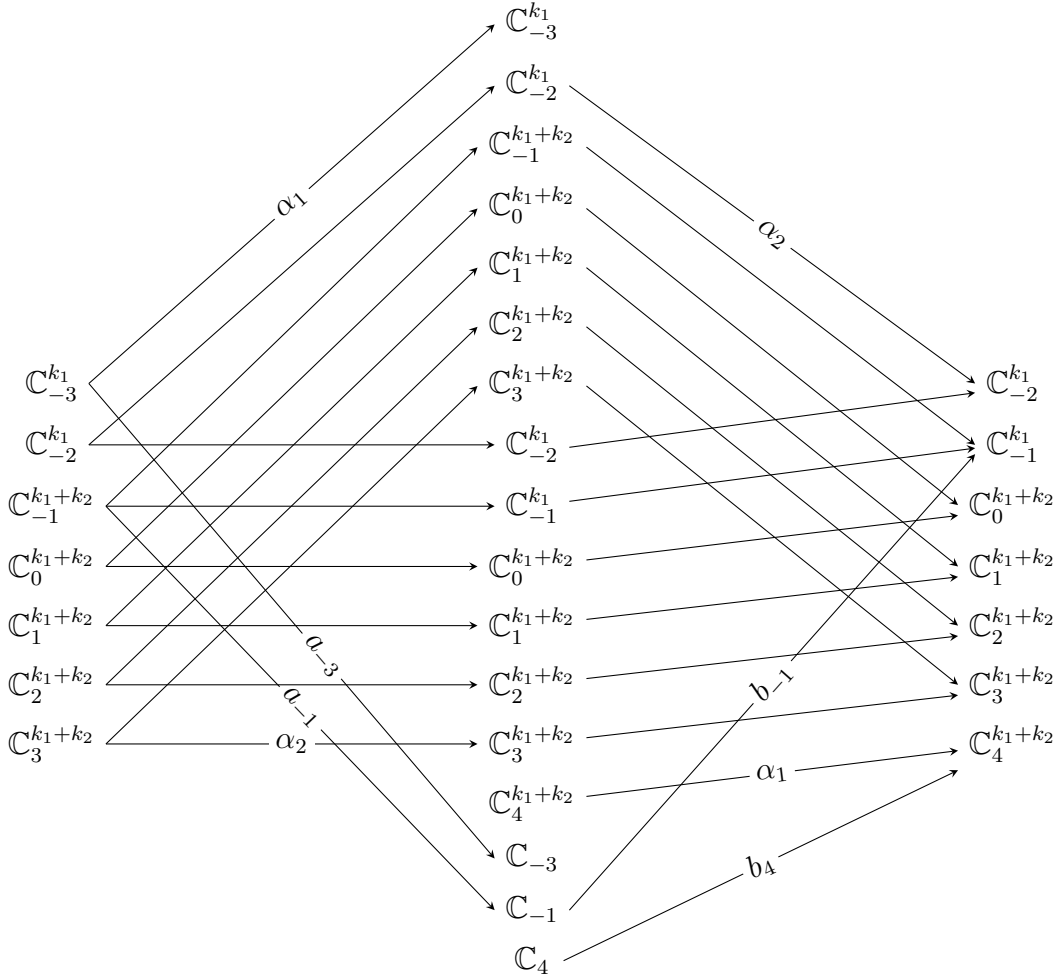

The matrix $\alpha_{1}$ is a sparse matrix with square blocks $\left\{ \beta_{i+1/2}\right\} ,\, p_{1}\leq i\leq p_{N}-1$
running down the diagonal of the indicated size. The matrix dimensions
increase from $\left(k_{1}+\ldots+k_{j-1}\right)\times\left(k_{1}+\ldots+k_{j-1}\right)$
to $\left(k_{1}+\ldots+k_{j}\right)\times\left(k_{1}+\ldots+k_{j}\right)$
at each $i=p_{j},\,2\leq j\leq N-1$. The subscripts of $\beta_{i+1/2}$,
$\gamma_{i}$, $a_{i}$ and $b_{i}$ indicate that they map between
spaces of weight $i$ of the $\mathbb{C}^{\times}$-action (between
$i$ and $i+1$ for the $\beta$s).


\begin{figure}
\begin{tikzpicture} 

\draw (0,0) rectangle (1,-1) 
            rectangle (2,-2) 
 (2.5,-2.5) rectangle (3.5,-3.5)
            rectangle (5,-5)
            rectangle (6.5,-6.5)
     (7,-7) rectangle (8.5,-8.5)
            rectangle (10,-10)
      (0,0) rectangle (10,-10);

\node[scale=0.8pt] at (0.5,-0.5) {$\beta_{p_1+\frac{1}{2}}$};
\node[scale=0.8pt] at (1.5,-1.5) {$\beta_{p_1+\frac{3}{2}}$};
\node[scale=0.8pt] at (2.25,-2.15) {$\ddots$};
\node[scale=0.8pt] at (3,-3) {$\beta_{p_i-\frac{1}{2}}$};
\node[scale=1pt] at (4.25,-4.25) {$\beta_{p_i+\frac{1}{2}}$};
\node[scale=1pt] at (5.75,-5.75) {$\beta_{p_i+\frac{3}{2}}$};
\node[scale=0.8pt] at (6.75,-6.65) {$\ddots$};
\node[scale=1pt] at (7.75,-7.75) {$\beta_{p_{N} -\frac{1}{2}}$};
\node[scale=1pt] at (9.25,-9.25) {$\beta_{p_{N}+\frac{1}{2}}$};
\node[scale=1.5] at (-1,-5.5) {$\alpha_1 =$};

\draw[decorate,decoration={brace,raise=6pt,amplitude=3pt}] (10,-0.1)--(10,-0.9);
\draw[decorate,decoration={brace,raise=6pt,amplitude=3pt}] (10,-1.1)--(10,-1.9);
\draw[decorate,decoration={brace,raise=6pt,amplitude=3pt}] (10,-2.6)--(10,-3.4);
\draw[decorate,decoration={brace,raise=6pt,amplitude=3pt}] (10,-3.6)--(10,-4.9);
\draw[decorate,decoration={brace,raise=6pt,amplitude=3pt}] (10,-5.1)--(10,-6.4);
\draw[decorate,decoration={brace,raise=6pt,amplitude=3pt}] (10,-7.1)--(10,-8.4);
\draw[decorate,decoration={brace,raise=6pt,amplitude=3pt}] (10,-8.6)--(10,-9.9);
\node at (11,-0.5) {$k_1$};
\node at (11,-1.5) {$k_1$};
\node at (11,-3) {$k_i-1$};
\node at (11,-4.25) {$k_i$};
\node at (11,-5.75) {$k_i$};
\node at (11,-7.75) {$-k_{N}$};
\node at (11,-9.25) {$-k_{N}$};

\end{tikzpicture}
\end{figure}


The sparse matrix $\alpha_{2}$ has (square except at transitions)
blocks $\{\gamma_{i}\},\,p_{1}+1\leq i\leq p_{N}-1$ along the super-diagonal.
At $i=p_{j},\,2\leq j\leq N-1$ , the diagonal block of zeros increases
in dimensions from $\left(k_{1}+\ldots+k_{j-1}\right)\times\left(k_{1}+\ldots+k_{j-1}\right)$
to $\left(k_{1}+\ldots+k_{j}\right)\times\left(k_{1}+\ldots+k_{j}\right)$.
The matrix $\gamma_{p_{j}}$ sitting in the transition is a \textit{rectangular}
matrix of dimensions $\left(k_{1}+\ldots+k_{j-1}\right)\times\left(k_{1}+\ldots+k_{j}\right)$.
The next matrix $\gamma_{p_{j}+1}$ returns to being a square block,
now of dimensions $\left(k_{1}+\ldots+k_{j}\right)\times\left(k_{1}+\ldots+k_{j}\right)$.

The $N\times\kappa$ matrix $a$ is divided by $P_{c}$ into columns
labelled by weight space. The non-zero entries are row vectors $\left\{ a_{1},\ldots,a_{N-1}\right\} $
in the columns with weight $p_{i}$, $1\leq i\leq N-1$ and $i$-th
rows of length $k_{1}+\ldots+k_{i}$. The last weight space of the
domain of $a$ correponding to the last $-k_{N}$ columns has weight
$p_{N}-1$.


\begin{figure}
\begin{tikzpicture} 

\draw (0,0) rectangle (1,-1) 
            rectangle (2,-2) 
 (2.5,-2.5) rectangle (3.5,-3.5)
            rectangle (5,-5)
            rectangle (6.5,-6.5)
     (7,-7) rectangle (8.5,-8.5)
            rectangle (10,-10)
      (1,0) rectangle (2,-1)
 (2.5,-1.5) rectangle (3.5,-2.5)
            rectangle (5,-3.5)
            rectangle (6.5,-6.5)
     (5,-5) rectangle (8,-6.5)
   (8.5,-7) rectangle (10,-8.5)
      (0,0) rectangle (10,-10);  

\node[scale=1pt] at (0.5,-0.5) {$0_{k_1}$};
\node[scale=1pt] at (1.5,-1.5) {$0_{k_1}$};
\node[scale=0.8pt] at (2.25,-2.15) {$\ddots$};
\node[scale=0.85pt] at (3,-3) {$0_{k_1}$};
\node[scale=1pt] at (4.25,-4.25) {$0_{k_1+k_2}$};
\node[scale=1pt] at (5.75,-5.75) {$0_{k_1+k_2}$};
\node[scale=0.8pt] at (6.75,-6.65) {$\ddots$};
\node[scale=1.5] at (-1,-5.5) {$\alpha_2 =$};
\node[scale=1pt] at (7.75,-7.75) {$0_{-k_{N}}$};
\node[scale=1pt] at (9.25,-9.25) {$0_{-k_{N}}$};

\node[scale=0.8pt] at (1.5,-0.5) {$\gamma_{p_1+1}$};
\node[scale=0.8pt] at (2.25,-1.15) {$\ddots$};
\node[scale=0.8pt] at (3,-2) {$\gamma_{p_i-1}$};
\node[scale=1pt] at (4.25,-3) {$\gamma_{p_i}$};
\node[scale=1pt] at (5.75,-4.25) {$\gamma_{p_i+1}$};
\node[scale=1pt] at (7.25,-5.75) {$\gamma_{p_i+2}$};
\node[scale=0.8pt] at (8.25,-6.65) {$\ddots$};
\node[scale=1pt] at (9.25,-7.75) {$\gamma_{p_{N}}$};

\draw[decorate,decoration={brace,raise=6pt,amplitude=3pt}] (10,-0.1)--(10,-0.9);
\draw[decorate,decoration={brace,raise=6pt,amplitude=3pt}] (10,-2.6)--(10,-3.4);
\draw[decorate,decoration={brace,raise=6pt,amplitude=3pt}] (10,-3.6)--(10,-4.9);
\draw[decorate,decoration={brace,raise=6pt,amplitude=3pt}] (10,-5.1)--(10,-6.4);
\draw[decorate,decoration={brace,raise=6pt,amplitude=3pt}] (10,-7.1)--(10,-8.4);
\draw[decorate,decoration={brace,raise=6pt,amplitude=3pt}] (10,-8.6)--(10,-9.9);
\node at (11,-0.5) {$k_1$};
\node at (11,-3) {$k_i-1$};
\node at (11,-4.25) {$k_i$};
\node at (11,-5.75) {$k_i$};
\node at (11,-7.75) {$-k_{N}$};
\node at (11,-9.25) {$-k_{N}$};

\end{tikzpicture}
\end{figure}



\begin{figure}
\begin{tikzpicture}

\draw (0,0) rectangle (1,-0.5)
   (2,-0.5) rectangle (3,-1)
     (4,-2) rectangle (5,-2.5)
   (6,-3.5) rectangle (7,-4)
     (8,-4) rectangle (9,-4.5);
\draw (0,0) rectangle (9,-4.5);

\draw[lightgray,ultra thin] (0,-0.5) -- (9,-0.5);
\draw[lightgray,ultra thin] (0,-1) -- (9,-1);
\draw[lightgray,ultra thin] (0,-2) -- (9,-2);
\draw[lightgray,ultra thin] (0,-2.5) -- (9,-2.5);
\draw[lightgray,ultra thin] (0,-3.5) -- (9,-3.5);
\draw[lightgray,ultra thin] (0,-4) -- (9,-4);

\node[scale=1.2] at (-1,-2.5) {$a =$};
\node at (3.5,-1.5) {$\ddots $};
\node at (5.5,-3) {$\ddots $};

\node[scale=0.9] at (0.5,-0.25) {$a_{p_1}$};
\node[scale=0.9] at (2.5,-0.75) {$a_{p_2}$};
\node[scale=0.9] at (4.5,-2.25) {$a_{p_i}$};
\node[scale=0.8] at (6.5,-3.75) {$a_{p_{\scriptscriptstyle{{N-1}}}}$};
\node[scale=0.8] at (8.5,-4.25) {$0_{p_{N}}$};

\draw[decorate,decoration={brace,raise=6pt,amplitude=3pt, mirror}] (0.1,-4.5)--(0.9,-4.5);
\draw[decorate,decoration={brace,raise=6pt,amplitude=3pt, mirror}] (2.1,-4.5)--(2.9,-4.5);
\draw[decorate,decoration={brace,raise=6pt,amplitude=3pt, mirror}] (4.1,-4.5)--(4.9,-4.5);
\draw[decorate,decoration={brace,raise=6pt,amplitude=3pt, mirror}] (6.1,-4.5)--(6.9,-4.5);
\node at (0.5,-5.3) {$k_1$};
\node at (2.5,-5.3) {$k_1+k_2$};
\node[scale=0.9pt] at (4.5,-5.3) {$k_1+\ldots+k_i$};
\node at (6.5,-5.3) {$-k_{N}$};

\end{tikzpicture}
\end{figure}


The $\kappa\times N$ matrix $b$ is divided into rows labelled by
weight space. The non-zero entries are column vectors $\left\{ b_{2},\ldots,b_{N}\right\} $
in the rows with weight $p_{i}$, $2\leq i\leq N-1$ and $p_{N}$,
and $i$-th columns of length $k_{1}+\ldots+k_{i-1}$. Note that the
first weight space of the image of $b$ corresponding to the first
$k_{1}$ rows has weight $p_{1}+1$.


\begin{figure}
\begin{tikzpicture}

\draw (0,0) rectangle (0.75,-1)
   (0.75,-1.5) rectangle (1.5,-2.5)
     (2.5,-3.5) rectangle (3.25,-4.5)
   (4.25,-5.5) rectangle (5.15,-6.5)
     (5.15,-7) rectangle (6,-8);
\draw (0,0) rectangle (6,-8);

\draw[lightgray,ultra thin] (0.75,0) -- (0.75,-8);
\draw[lightgray,ultra thin] (1.5,0) -- (1.5,-8);
\draw[lightgray,ultra thin] (2.5,0) -- (2.5,-8);
\draw[lightgray,ultra thin] (3.25,0) -- (3.25,-8);
\draw[lightgray,ultra thin] (4.25,0) -- (4.25,-8);
\draw[lightgray,ultra thin] (5.15,0) -- (5.15,-8);

\node[scale=1.2] at (-1,-4) {$b =$};
\node at (2,-3) {$\ddots $};
\node at (3.75,-5) {$\ddots $};

\node[scale=0.9] at (0.4,-0.5) {$0_{k_1}$};
\node[scale=0.9] at (1.15,-2) {$b_{p_2}$};
\node[scale=0.9] at (2.9,-4) {$b_{p_i}$};
\node[scale=0.75] at (4.7,-6) {$b_{p_{\scriptscriptstyle{N-1}}}$};
\node[scale=0.8] at (5.65,-7.5) {$b_{p_{N}}$};

\draw[decorate,decoration={brace,raise=6pt,amplitude=3pt}] (6,-1.6)--(6,-2.4);
\draw[decorate,decoration={brace,raise=6pt,amplitude=3pt}] (6,-3.6)--(6,-4.4);
\draw[decorate,decoration={brace,raise=6pt,amplitude=3pt}] (6,-5.6)--(6,-6.4);
\draw[decorate,decoration={brace,raise=6pt,amplitude=3pt}] (6,-7.1)--(6,-7.9);
\node at (6.75,-2) {$k_1$};
\node at (7.75,-4) {$k_1+\ldots+k_{i-1}$};
\node at (7.75,-6) {$k_1+\ldots+k_{N-2}$};
\node at (6.85,-7.5) {$-k_N$};

\end{tikzpicture}
\end{figure}


The complex equation (\ref{eq:complex ADHM}) is now a series of equations
in terms of the blocks $\left\{ \beta_{i+1/2}\right\} _{p_{1}\leq i\leq p_{N}-1}$
and $\left\{ \gamma_{j}\right\} _{p_{1}+1\leq j\leq p_{N}-1}$,

\begin{equation}
\begin{cases}
\beta_{i+\frac{1}{2}}\gamma_{i+1}-\gamma_{i+1}\beta_{i+\frac{3}{2}}+b_{i+1}a_{i+1}=0 & \text{for }i+1=p_{j},\,2\leq j\leq N-1\\
\beta_{i+\frac{1}{2}}\gamma_{i+1}-\gamma_{i+1}\beta_{i+\frac{3}{2}}=0 & \text{otherwise}
\end{cases}\label{eq:DN complex}
\end{equation}
which we call the complex discrete Nahm equations.

The real ADHM equation becomes the real discrete Nahm equations

\begin{equation}
\begin{cases}
\left[\beta_{i+\frac{1}{2}},\beta_{i+\frac{1}{2}}^{\ast}\right]+\gamma_{i+1}\gamma_{i+1}^{\ast}-\gamma_{i}^{\ast}\gamma_{i}-a_{i}^{\ast}a_{i}=0 & \text{when }i=p_{j},\,1\leq j\leq N-1\\
\left[\beta_{i+\frac{1}{2}},\beta_{i+\frac{1}{2}}^{\ast}\right]+\gamma_{i+1}\gamma_{i+1}^{\ast}-\gamma_{i}^{\ast}\gamma_{i}+b_{i+1}b_{i+1}^{\ast}=0   & \text{when }i+1=p_{j},\,2\leq j\leq N \\ 
\left[\beta_{i+\frac{1}{2}},\beta_{i+\frac{1}{2}}^{\ast}\right]+\gamma_{i+1}\gamma_{i+1}^{\ast}-\gamma_{i}^{\ast}\gamma_{i}=0 & \text{otherwise}
\end{cases}
\end{equation}
where $\gamma_{p_{1}}=0=\gamma_{p_{N}}$ so the first real equation
is 
\[
\left[\beta_{p_{1}+\frac{1}{2}},\beta_{p_{1}+\frac{1}{2}}^{\ast}\right]+\gamma_{p_{1}+1}\gamma_{p_{1}+1}^{\ast}-a_{p_{1}}^{\ast}a_{p_{1}}=0
\]
and the last one is 
\[
\left[\beta_{p_{N}-\frac{1}{2}},\beta_{p_{N}-\frac{1}{2}}^{\ast}\right]+b_{p_{N}+\frac{1}{2}}b_{p_{N}+\frac{1}{2}}^{\ast}-\gamma_{p_{N}-1}^{\ast}\gamma_{p_{N}-1}=0.
\]

\begin{defn}
A solution of the $(N-1)$-interval discrete Nahm equations of type
$(p_{1},\ldots,p_{N-1};k_{1},\ldots,k_{N-1})\in\mathbb{Z}^{2(N-1)}$
is a equivalence class of matrices 
\[
\left(\{\beta_{j}\},\{\gamma_{j}\},\{a_{p_{i}}\},\{b_{p_{i}}\}\right)
\]
labeled by half-integer points on an interval $j\in[p_{1},p_{N}]$
as shown 


\begin{figure}[H]
\begin{tikzpicture}[>= latex,descr/.style={fill=white,inner sep=2.5pt}] 
\draw[-] (0,0) -- (4,0) (5,0) -- (9,0) (10,0) -- (12,0);
\draw[loosely dotted, thick] (4,0) -- (5,0) (9,0) -- (10,0);

\foreach \x/\xtext in {0/{p_1}, 1/, 2/{p_1+1}, 3/, 4/{p_1+2}, 5/{p_2-1}, 6/, 7/{p_2}, 8/, 9/{p_2+1},10/{p_N-1},11/,12/{p_N}}
\draw[shift={(\x,0)}] (0pt,2pt) -- (0pt,-2pt) node[below] {$\scriptstyle{\xtext}$};

\foreach \x/\xtext in {0/a, 1/\beta, 2/\gamma, 3/\beta, 4/\gamma, 5/\gamma, 6/\beta, 7/{b,\gamma, a}, 8/\beta, 9/\gamma,10/\gamma,11/\beta,12/b}
\draw[shift={(\x,0)}] (0pt,10pt) -- (0pt,10pt) node[above] {$\xtext$};
\end{tikzpicture}
\end{figure}


\begin{flushleft} with dimensions $(k_{1}+\ldots+k_{i})\times(k_{1}+\ldots+k_{i})$
 at half integer points on an interval $(p_{i},p_{i+1})$ and at a boundary point $p_{i}$
between intervals, the matrices $a_{p_{i}}$ , $\gamma_{p_{i}}$
and $b_{p_{i}}$ have dimensions $1\times(k_{1}+\ldots+k_{i})$,
$(k_{1}+\ldots+k_{i-1})\times(k_{1}+\ldots+k_{i})$ and $(k_{1}+\ldots+k_{i-1})\times1$
respectively. The matrices must satisfy the $(N-1)$-interval discrete
Nahm equations and satisfy the equivalence relation (``gauge transformations'') \end{flushleft}
\vspace{0.5cm} 
\[
\begin{aligned}\beta_{j} & \sim g_{j}\beta_{j}g_{j}^{-1}\\
\gamma_{j} & \sim g_{j-\frac{1}{2}}\gamma_{j}g_{j+\frac{1}{2}}\\
a_{p_{i}} & \mapsto\lambda_{p_{i}}a_{p_{i}}g_{p_{i}+\frac{1}{2}}^{-1}\\
b_{p_{i}} & \mapsto g_{p_{i}-\frac{1}{2}}b_{p_{i}}\lambda_{p_{i}}^{-1}
\end{aligned}
\]
\vspace{0.5cm} 
where $g_{j}\in\text{U}(k_{1}+\ldots+k_{i})$ when $j\in(p_{i},p_{i+1})$. 
\end{defn}

Thus is our first main theorem proven: 
\begin{thm}
There is an equivalence between

\vspace{0.5cm}

\begin{minipage}[t]{1\columnwidth}%
\begin{enumerate}
\item framed $\text{SU}(N)$ monopoles $(A,\phi)$ on hyperbolic space $H^{3}$
of mass $(p_{i},\ldots,p_{N-1})\in\mathbb{Z}^{N-1}$ (or $\left(\frac{1}{2}+\mathbb{Z}\right)^{N-1}$
for $N$ even) and charge $(k_{1},\ldots,k_{N-1})\in\mathbb{Z}^{N-1}$,
and 

\vspace{0.5cm}

\item solutions of the $(N-1)$-interval discrete Nahm equations of type
$(p_{1},\ldots,p_{N-1};k_{1},\ldots,k_{N-1})$.\end{enumerate}
\end{minipage} 
\end{thm}

\bigskip{}

\pagebreak


\section{The rational map}

Atiyah \cite{key-5} showed that:
\begin{thm}[Atiyah] 
 For a compact classical group $G$, the moduli space of circle-invariant
instantons or equivalently, hyperbolic monopoles of charge \\ $\boldsymbol{k}=\left(k_{1},\ldots,k_{N}\right)$
is isomorphic to the space of degree $\boldsymbol{k}$ ``rational
maps'' 
\[
f:\mathbb{P}^{1}\rightarrow G/T
\]
where $T$ is a maximal torus. 
\end{thm}
When $G=\text{SU}(N)$, $G/T=\text{Fl}_{\text{full}}(N)=\left\{ 0\subset\mathbb{C}\subset\mathbb{C}^{2}\subset\ldots\subset\mathbb{C}^{N}\right\} $,
the manifold of full flags in $N$-dimensional space. For magnetic
monopoles, we have the following corollary. 
\begin{cor}
There is an isomorphism between the moduli of framed $\text{SU}(N)$
magnetic monopoles on $H^{3}$ and the moduli of degree $(k_{1},k_{1}+k_{2},\ldots,k_{1}+\ldots+k_{N-1})$
rational maps such that $f(\infty)=\boldsymbol{0}$, 
\[
f:\mathbb{P}^{1}\rightarrow\text{Fl}_{\text{full}}(N).
\]

\end{cor}
Along the lines of Braam and Austin \cite{key-3}, I will derive an
explicit formula for the rational map of a hyperbolic monopole in
terms of its discrete Nahm boundary data. To do this, restrict the
bundle to the projective plane $\mathbb{P}^{2}=\{[x:y:z:0]\in\mathbb{P}^{3}\}$.
Over this $\mathbb{P}^{2}$, the solutions of the discrete Nahm equations
have a $\text{GL}(\boldsymbol{k},\mathbb{C})$ freedom. We first require
two lemmas of Braam and Austin whose conditions are satisfied in our
case. 
\begin{lem}[Braam-Austin 4.2]
 \label{lem:BA 4.2}If $(\{\gamma_{i}\},\{\beta_{i}\},\{a_{p_{j}}\},\{b_{p_{j+1}}\})$
lies in a stable orbit then the $\gamma_{i}$ are all injective. 
\end{lem}
By the injectivity of the $\gamma_{i}$ and using the $ $$\text{GL}(\boldsymbol{k},\mathbb{C})$
action, 
\[
g_{i-\frac{1}{2}}\gamma_{i}g_{i+\frac{1}{2}}^{-1}=\text{I}
\]
we set all the interval $\gamma_{i}$ to the identity matrix. Then
in each interval, the $\beta_{i}$ are all equal to constant matrix
$\beta_{[p_{i}]}$ with subscript labelling the boundary point before
the interval. Square brackets in the subscript indicate that this
is the matrix after the $\text{GL}(\boldsymbol{k},\mathbb{C})$ action
has been applied. 
\begin{lem}[Braam-Austin 4.3]
 The data $(\{\beta_{[p_{i}]}\},\{\gamma_{[p_{i}]}\},\{a_{[p_{i}]}\},\{b_{[p_{i+1}]}\})$
defines a monad satisfying the ADHM equations if and only if $\{\beta_{[p_{i}]}^{l}a_{[p_{i}]}\}$
for $l=0,\ldots,k_{1}+\ldots+k_{i}$ span $\mathbb{C}^{k_{1}+\ldots+k_{i}}$.
\end{lem}
The procedure is as follows. Choose a ``horosphere line'' $\mathbb{P}_{h}^{1}$
in $\mathbb{P}^{2}$ with coordinates say $x\mapsto[x:h:-1]$. The
trivialisation of $E$ over $\mathbb{P}_{\infty}^{1}$ is also a trivialisation
of the monad in the sense that over $\mathbb{P}_{\infty}^{1}$, $(\boldsymbol{0},\boldsymbol{0},r)\in K$,
$r\in\mathbb{C}^{N}$ are representatives of the global sections of
$E\vert_{\mathbb{P}_{\infty}^{1}}$. Extended to $\mathbb{P}_{h}^{1}$,
this trivialisation is 
\[
\left[\begin{array}{c}
-\left(h-\alpha_{2}\right)^{-1}b\\
0_{\kappa}\\
I_{N}
\end{array}\right]r+\left[\begin{array}{c}
\left(h-\alpha_{2}\right)^{-1}\left(x-\alpha_{1}\right)\\
I_{\kappa}\\
0_{N}
\end{array}\right]Y\in K
\]
where $Y\in\mathbb{C}^{\kappa}$.

Consider the splitting of $E$ over $\mathbb{P}_{+}^{1}$, 
\[
E=\mathcal{O}(k_{1})\otimes\mathcal{L}^{p_{1}}\oplus\ldots\oplus\mathcal{O}(k_{r})\otimes\mathcal{L}^{p_{r}}\oplus\ldots\oplus\mathcal{O}\left(k_{N}\right)\mathcal{L}^{p_{N}}.
\]
Atiyah showed that in the $\text{SU}(2)$ case, the last factor extends
by flowing along the $\mathbb{C}^{\times}$-action to a sub-line-bundle
over $\mathbb{P}^{3}-\mathbb{P}_{-}^{1}$. The sum of the last two
factors extend to a sub-plane-bundle and the sum of the last three
extend to a rank 3 sub-bundle of $E$, etc. 

\begin{lem}
On $\mathbb{P}^{2}-\mathbb{P}_{-}^{1}$, there exists unique holomorphic
sub-bundles $L_{1}^{+}\subset L_{2}^{+}\subset\ldots\subset L_{N-1}^{+}$
of $E$ which is preserved by the $\mathbb{C}^{\times}$-action and
each $L_{i}^{+}$ restricted to $\mathbb{P}_{+}^{1}$ coincides with
the last $i$-th factors.
\end{lem}

\begin{proof}
The bundle $E$ restricted to a $\mathbb{C}^{\times}$-orbit $\mathbb{P}^{1}-\{\text{pt of }\mathbb{P}_{-}^{1}\}$
has the following $\mathbb{C}^{\times}$-action:

\[
c\cdot(z;u_{1},\ldots,u_{N})=\left(cz;c^{p_{1}}u_{1},\ldots,c^{p_{N}}u_{N}\right).
\]
In the limit $c\rightarrow0$, the global holomorphic sections of
the form $(0,0,\ldots,0,u_{N}(z))$ are preserved by the $\mathbb{C}^{\times}$-action
since multiplication by $c\in\mathbb{C}^{\times}$ cannot change zero
into a non-zero number. Since the space of such sections is one dimensional,
they give us a sub-line bundle $L_{1}^{+}$ of $E$. The sections
have weight $-p_{N}$ and so must coincide with the first factor in
the splitting of $E$ over $\mathbb{P}_{+}^{1}$.

Similarly for $1<i<N$, in the $c\rightarrow0$ limit, the global
holomorphic sections 
\[
(0,\ldots,0,u_{i}(z),u_{i+1}(z),\ldots,u_{N}(z)),
\]
are preserved by the $\mathbb{C}^{\times}$-action and have weights
$(p_{i},\ldots,p_{N})$. The set of them is $\left(N-i+1\right)$-dimensional
so they define a rank $(N-i+1)$ sub-bundle $L_{N-i+1}^{+}$ of $E$.

By induction, a section of the form $(0,\ldots,0,u_{i}(z),\ldots,u_{N}(z))$
is also a section of the sub-bundle given by sections of the form
$(0,\ldots,u_{i-1}(z),\ldots,u_{N}(z))$ so $L_{N-i+1}^{+}\subset L_{N-i}^{+}$
and thus the sub-bundles are a chain ordered by subset.

These are the only sections preserved by the $\mathbb{C}^{\times}$-action
since the $\mathbb{C}^{\times}$-action is transitive on the non-zero
entries of sections. Hence the holomorphic sub-bundles $L_{1}^{+}\subset\ldots\subset L_{N-1}^{+}$
preserved by the $\mathbb{C}^{\times}$-action thus defined are unique. 
\end{proof}

The rational map $f$ is defined by sending each point $x$ of $\mathbb{P}_{+}^{1}$
to the fibre of the restriction of $L_{1}^{+}\subset\ldots\subset L_{N-1}^{+}\subset E$
to the orbit of $\mathbb{C}^{\times}$ whose limit is $x$. The chain
of sub-bundles over the $\mathbb{C}^{\times}$-orbit is trivialised
by taking the intersection of the $\mathbb{C}^{\times}$-orbit with
the chosen horosphere line $\mathbb{P}_{h}^{1}$ as the unit point
and then the rest of the isomorphism is constructed by flowing along
the $\mathbb{C}^{\times}$-orbit using the $\mathbb{C}^{\times}$-action.
Canonically, 
\[
\left(L_{1}^{+},\ldots,L_{N-1}^{+}\right)\vert_{\mathbb{C}^{\times}}\cong(\mathbb{C}^{1},\ldots,\mathbb{C}^{N-1})\times\mathbb{C}^{\times}
\]
so that $f(z)$ is an element of the manifold of full flags $\text{Fl}_{\text{full}}(N)$.

Since $E$ has a canonical trivialisation over $\mathbb{P}_{h}^{1}$,
we can find equations for the rational map. On the level of the monad,
the rank $i$ sub-bundle is produced exactly when the $p_{1},\ldots,p_{N-i}$
weight spaces are in the kernel of $A_{X}$. This happens when the
expression for each $p_{i}$ weight space in the monad trivialisation
is equal to the negative of some element of the image of $A_{X}$.

Using Lemma \ref{lem:BA 4.2} to linearly transform $\{\gamma_{[j]}\}_{j\neq p_{i}}$
into identity matrices, we can invert $(h-\alpha_{2})$. Writing $r=(r_{1},\ldots,r_{N})$,
we define the algebraic equations of a flag of subspaces by recursion.
The condition that the $p_{1}$ weight space be in the kernel of
$A_{X}$ is equivalent to solving the equations 
\[
(-h)^{p_{N-1}-p_{N}}b_{[p_{N}]}r_{N}+(x-\beta_{[p_{N-1}+\frac{1}{2}]})w_{p_{N-1}}=0
\]
\[
r_{N-1}+a_{[p_{N-1}]}w_{p_{N-1}}=0.
\]
Solving for $r_{N-1}$ in terms of $r_{N}$, this is 
\[
r_{N-1}=(-h)^{p_{N-1}-p_{N}}a_{[p_{N-1}]}\left(x-\beta_{[p_{N-1}]}\right)^{-1}b_{[p_{N}]}r_{N}
\]
which defines a line in a plane for any $x\in\mathbb{P}^{1}$.

Proceeding in the same way for the other weight spaces, we have: 
\begin{prop}
Let $(\{\gamma_{i}\},\{\beta_{i}\},\{a_{p_{j}}\},\{b_{p_{j+1}}\})$
be a solution of the $(N-1)$-interval discrete Nahm equations of
type $\left(p_{1},\ldots,p_{N-1};k_{1},\ldots,k_{N-1}\right)$.
Then the solution can be put into the form $(\{\beta_{[p_{i}]}\},\{\gamma_{[p_{i}]}\},\{a_{[p_{i}]}\},\{b_{[p_{i+1}]}\})$
and the rational map,

\[
f:\mathbb{P}^{1}\rightarrow\text{Fl}_{\text{full}}(N)
\]
\[
x\mapsto(V_{1},\ldots,V_{N-1}),\ \ \text{dim }V_{i}=i,
\]

into the manifold of full flags in $\mathbb{C}^{N}$ can be written
as the maps $(r_{1}(x),\ldots,r_{N-1}(x))$,

\begin{eqnarray*}
r_{N-1}(x) & = & (-h)^{p_{N-1}-p_{N}}a_{[p_{N-1}]}\left(x-\beta_{[p_{N-1}]}\right)^{-1}b_{[p_{N}]}r_{N}(x)\\
 & \vdots\\
r_{j}(x) & = & \sum_{i=j+1}^{N}(-h)^{p_{j}-p_{i}}a_{[p_{j}]}\left(x-\beta_{[p_{j}]}\right)^{-1}b_{[p_{i}]}^{k_{1}+\ldots+k_{j}}r_{i}(x)\\
 & \vdots\\
r_{1}(x) & = & \sum_{i=2}^{N}(-h)^{p_{1}-p_{i}}a_{[p_{1}]}\left(x-\beta_{[p_{1}]}\right)^{-1}b_{[p_{i}]}^{k_{1}}r_{i}(x)
\end{eqnarray*}
where for each $x\in\mathbb{P}^{1}$, $r_{N-1}(x)$ specifies an $(N-1)$-dimensional
linear subspace in $\mathbb{C}^{N}$ and each successive $r_{i}(x)$
specifies an $i$-dimensional linear subspace inside the $(i+1)$-dimensional
linear subspace specified by $r_{i+1}(x)$. The superscript $k_{1}+\ldots+k_{j}$
indicates that only the first $k_{1}+\ldots+k_{j}$ entries of the
vector are involved. 
\end{prop}

\bigskip{}
Note that when $N=2,$ the equation of the rational map is of the
form 
\[
r(x)=\dfrac{r_{2}(x)}{r_{1}(x)}=(-h)^{2p}v(x-\beta)^{-1}v^{t}
\]
which is the rational map found by Atiyah for $\text{SU}(2)$ hyperbolic
monopoles \cite{key-3,key-4}.


\section{The Boundary Value of a Monopole}

On the conformal sphere at infinity, $S_{\infty}^{2}$, the holomorphic
vector bundle $\mathcal{E}$ splits into holomorphic line bundles
$\mathcal{O}(k_{1})\oplus\ldots\oplus\mathcal{O}(k_{N-1})$ and the
gauge field $A$ restricted to $S_{\infty}^{2}$, induces a a $\text{U}(1)$
connection $A_{i}$ on each factor $\mathcal{O}(k_{i})$. We define
the $(N-1)$-tuple $(A_{1},\ldots,A_{N-1})$ to be the boundary value
or connections at infinity.

We shall prove the following generalisation of Braam-Austin's theorem
\cite{key-3} regarding the boundary values of $\text{SU}(2)$ hyperbolic
monopoles. 

\begin{thm}
Let $(A,\Phi)$ be a framed $\text{SU}(2)$ hyperbolic monopole. Then

\begin{minipage}[t]{1\columnwidth}%
\begin{enumerate}
\item the $(N-1)$ tuple of $\text{U}(1)$ connections $(A_{1},\ldots,A_{N-1})$
on $S_{\infty}^{2}$ determines the connection $A$ (up to gauge transformations); 
\item there exists for $i=1,\ldots,N-1$, holomorphic maps 
\[
F_{i}:\mathbb{P}^{1}\rightarrow\text{Fl}(k_{1}+\ldots+k_{i},k_{1}+\ldots+k_{i}+1,2k_{1}+\ldots+2k_{i-1}+k_{i}+1)
\]
into the manifold of two term partial flags for which each $A_{i}$
is the pullback of the unitary invariant connection on the ``hyperplane
bundle'' $\mathcal{O}(1,-1)$ of the $i$-th flag manifold; and 
\item the map $A\mapsto(A_{1},\ldots,A_{N-1})$ is an immersion of the moduli
space of $\text{SU}(N)$ framed hyperbolic monopoles in the moduli
of $(N-1)$ tuples of $\text{U}(1)$ connections on $S^{2}$.\end{enumerate}
\end{minipage}\end{thm}

\begin{proof}
From Lemma, we have a decomposition of the monad $H\rightarrow K\rightarrow L$
restricted to $\mathbb{P}_{+}^{1}$ (which by abuse of notation, I
conflate with $S_{\infty}^{2}$ since any connections on $\mathbb{P}_{+}^{1}$
descend to connections on $S_{\infty}^{1}$ along the twistor transform)
into weight spaces. By considering the maps $A_{x}$ and $B_{x}$
restricted to a weight subspace, we get what is called a small monad.
By dimensional considerations, the cohomology of a generic small monad
($p_{i}<j<p_{i+1}$) 
\[
\xymatrix{ & \mathbb{C}^{k_{1}+\ldots+k_{i}}\ar[r]^{\gamma_{j}} & \mathbb{C}^{k_{1}+\ldots+k_{i-1}}\\
\mathbb{C}^{k_{1}+\ldots+k_{i}}\ar[r]_{\gamma_{j}}\ar[ur]^{\beta_{j+\frac{1}{2}}} & \mathbb{C}^{k_{1}+\ldots+k_{i-1}}\ar[ur]_{\beta_{j-\frac{1}{2}}}
}
\]
is trivial except for the weight spaces $p_{1},\ldots,p_{N}$
which take the form

\[
\xymatrix{ & \mathbb{C}_{p_{i}}^{k_{1}+\ldots+k_{i}}\ar[rd]^{\gamma_{p_{i}}}\\
\mathbb{C}_{p_{i}}^{k_{1}+\ldots+k_{i}}\ar[ru]^{\beta_{p_{i}+\frac{1}{2}}}\ar[r]^{\gamma_{p_{i}}}\ar[dr]_{a_{p_{i}}} & \mathbb{C}_{p_{i}}^{k_{1}+\ldots+k_{i-1}}\ar[r]^{\beta_{p_{i}-\frac{1}{2}}} & \mathbb{C}_{p_{i}}^{k_{1}+\ldots+k_{i-1}}\\
 & \mathbb{\mathbb{C}}_{p_{i}}\ar[ur]_{b_{p_{i}}}
}
\]

The cohomology of these small monads are holomorphic line bundles
defined fibre-wise

\[
L_{p_{i}}(x)=\text{ker}(\mathbb{C}^{2k_{1}+\ldots+2k_{i-1}+k_{i}+1}\rightarrow\mathbb{C}^{k_{1}+\ldots+k_{i-1}})/A_{x}(\mathbb{C}^{k_{1}+\ldots+k_{i}})
\]
which are exactly the line bundles in the splitting of $\mathcal{E}$.

Furthermore, there is a natural interpretation of the maps $A_{x}$
and $B_{x}$, restricted to each weight space of weight $p_{i}$
as a pair of maps,

\[
B_{x}^{t}:\mathbb{C}^{k_{1}+\ldots+k_{i-1}}\rightarrow\mathbb{C}^{2k_{1}+\ldots+2k_{i-1}+k_{i}+1}
\]

\[
A_{x}:\mathbb{C}^{k_{1}+\ldots+k_{i}}\rightarrow B_{x}^{t}(\mathbb{C}^{k_{1}+\ldots+k_{i-1}})^{\perp}\cong\mathbb{C}^{k_{1}+\ldots+k_{i}+1}\subset\mathbb{C}^{2k_{1}+\ldots+2k_{i-1}+k_{i}+1}
\]
defining a map $F_{i}=(A_{x}(H_{p_{i}}),B_{x}(L_{p_{i}})^{\perp})$
into the two term partial flag manifold $\text{Fl}(k_{1}+\ldots+k_{i},k_{1}+\ldots+k_{i}+1,2k_{1}+\ldots+2k_{i-1}+k_{i}+1)$.
Then each line bundle $L_{p_{i}}$ and its $\text{U}(1)$ connection
is the pullback of the invariant line bundle and connection over the
two term partial flag manifold. This proves (2) of the theorem.

The map $F_{i}$ thus defined is an embedding of $\mathbb{P}^{1}$
into the partial flag manifold, for the ADHM equations guarantee that
the monad is non-degenerate \cite{key-8}, and so $\text{im }F_{i}$
has no self-intersections and its derivative is non-zero. Compose
$F_{i}$ with the Plücker embedding and then the Segre embedding to
get 
\[
F_{i}^{\mathbb{P}}:\mathbb{P}^{1}\hookrightarrow\mathbb{P}^{\mathfrak{k}(i)}
\]
where 
\[
\mathfrak{k}(i)=\left(\begin{array}{c}
{\scriptstyle 2k_{1}+\ldots+2k_{i-1}+k_{i}+1}\\
{\scriptstyle k_{1}+\ldots+k_{i}}
\end{array}\right)\left(\begin{array}{c}
{\scriptstyle 2k_{1}+\ldots+2k_{i-1}+k_{i}+1}\\
{\scriptstyle k_{1}+\ldots+k_{i}+1}
\end{array}\right)-1.
\]

The pullback of the $\text{U}(\mathfrak{k}(i)+1)$ invariant connection
$A_{i}$ by the embedding $F_{i}^{\mathbb{P}}$ induces a Kähler form
$F_{A_{i}}$ (the curvature form of $A_{i}$) on $\mathbb{P}^{1}$.
The work of Calabi \cite{key-13} tells us that any such embedding
$F_{i}^{\mathbb{P}}$ is locally rigid, that is, the embedding is
determined by the Kähler form up to the isometry group of the target
space.

Hence the boundary values $(A_{1},\ldots,A_{N-1})$ descend by the
twistor transform to $\text{U}(1)$ connections on $S^{1}$ and determine
the small monad for the weight spaces corresponding to the weights
$p_{1}$,$\ldots,$$p_{N-1}$. These small monads provide boundary
values for the $(N-1)$-interval discrete Nahm equations and their
propagation uniquely specifies a complete solution up to gauge transformations.
Thus the boundary values on $S_{\infty}^{1}$ or equivalently $\mathbb{P}_{+}^{1}$
uniquely determine the monopole.

On the moduli space of $\text{SU}(N)$ framed hyperbolic monopoles,
the boundary values $(A_{1},\ldots,A_{N-1})$ are local coordinates.
Thus $A\mapsto(A_{1},\ldots,A_{N-1})$ is a local immersion of the
moduli of monopoles into the moduli of $(N-1)$-tuples of $\text{U}(1)$
connections on $S^{1}$. 
\end{proof}


\section{Final Remarks}

I have shown in this paper that 

\begin{minipage}[t]{1\columnwidth}%
\begin{enumerate}
\item There is an equivalence between framed $\text{SU}(N)$ hyperbolic
monopoles $(A,\phi)$ of charge $(p_{1},\ldots,p_{N-1})$ and charge
$(k_{1},\ldots,k_{N-1})$, and solutions $(\{\beta_{i}\},\{\gamma_{i}\},\{a_{p_{j}}\},\{b_{p_{i}}\})$
of the $(N-1)$-interval discrete Nahm equations of type $(p_{1},\ldots,p_{N-1};k_{1},\ldots,k_{N-1})$;
\item The rational map $\mathbb{P}^{1}\rightarrow\text{SU}(N)/U(1)^{N-1}$
of a hyperbolic monopole can be written explicitly from a solution
$(\{\beta_{j}\},\{\gamma_{j}\},\{a_{p_{i}}\},\{b_{p_{i}}\})$ of
the discrete Nahm equations; and
\item An $\text{SU}(N)$ hyperbolic monopole $(A,\phi)$ is determined by
its boundary value $(N-1)$-tuple of $\text{U}(1)$ connections $(A_{1},\ldots,A_{N-1})$
on the conformal boundary sphere $\mathbb{P}_{\infty}^{1}$ of $H^{3}$.\end{enumerate}
\end{minipage}

\bigskip{}

Note that the $(N-1)$-interval discrete Nahm equations are are essentially
$(N-1)$ copies of the ($\text{SU}(2)$) discrete Nahm equations linked
by an equation of the form 
\[
\beta_{p_{i}-\frac{1}{2}}\gamma_{p_{i}}-\gamma_{p_{i}}\beta_{p_{i}+\frac{1}{2}}+b_{p_{i}}a_{p_{i}}=0.
\]
It is interesting to interpret the $(N-1)$-interval discrete Nahm
equations as a representation of the type $A$ quiver. The $\text{A}_{N-1}$
Dynkin quiver diagram is the directed graph 
\[
\xymatrix{\circ\ar@{-}[r]\ar[r(0.6)]\ar@(ul,ur) & \circ\ar@{-}[r]\ar[r(0.6)]\ar@(ul,ur) & \circ\ar@(ul,ur)\ar@{-}[r]\ar[r(0.6)] & \ldots &\circ\ar@{-}[r]\ar[r(0.7)]\ar@(ul,ur) & \circ\ar@(ul,ur)}
\]
with $(N-1)$ vertices. Associate a vector space $V_{i}\simeq\mathbb{C}^{k_{1}+\ldots+k_{i}}$
to the $i$-th vertex, the operator $\beta_{[p_{i}]}:V_{i}\rightarrow V_{i}$
(from Section 4) to each curved arrow and the operator $b_{[p_{i}]}a_{[p_{i}]}:V_{i}\rightarrow V_{i+1}$
to each edge between distinct vertices. This is one way that the $(N-1)$-interval
discrete Nahm equations could be treated as a representation of the
$A_{N-1}$ quiver. It would be interesting to study such representations
and their relationship with such quiver representations appearing
in the study of supersymmetric models whose Coulomb branches involve
BPS monopoles \cite{key-14}. 

Another interesting avenue of research would be to study the spectral
curve associated to $\text{SU}(N)$ hyperbolic monopoles in terms
of the $(N-1)$-interval discrete Nahm equations. This is being studied
in on-going work with M.K. Murray. It is known that spectral data
does determine the monopole for a (apparently) different set of decay
conditions \cite{key-15}.

\bigskip{}

\paragraph*{\textit{Acknowledgements}}

I would like to thank my PhD supervisor, Paul Norbury for suggesting
the project and for his guidance. I acknowledge the Australian Department
of Education and Training for providing PhD funding via the Australian
Postgraduate Award.


\end{document}